\documentclass[a4paper, 12pt]{article}
\usepackage{amssymb, amsmath, amsthm}
\usepackage{graphicx}
\usepackage{caption}
\usepackage{subcaption}

\newcommand\R{\mathbb{R}}
\newcommand\C{\mathbb{C}}
\newcommand\N{\mathbb{N}}
\newcommand\B{\mathbb{B}}

\newcommand\1{{1\!\!1}}
\newcommand\KK{{\mathrm{K}}}

\newcommand\vd{\delta}
\newcommand\eps{\varepsilon}

\newcommand\llvert{\lvert\!\!\lvert}
\newcommand\rrvert{\rvert\!\!\rvert}

\newtheorem{theorem}{Theorem}[section]

\newtheorem{lemma}[theorem]{Lemma}
\newtheorem{proposition}[theorem]{Proposition}
\newtheorem{definition}[theorem]{Definition}

\newtheorem{remark}[theorem]{Remark}

\numberwithin{equation}{section}
\allowdisplaybreaks[3]

\begin{document}

\author{{\bf Dmitri Finkelshtein}\\
    {\small Department of Mathematics, Swansea University,}\\
    {\small Singleton~Park,~Swansea~SA2~8PP,~U.K.;}\\
    {\small Institute of Mathematics, Kiev, Ukraine}\\
    {\small d.l.finkelshtein@swansea.ac.uk}
\and
    {\bf Yuri Kondratiev}\\
    {\small Fakult\"at f\"ur Mathematik, Universit\"at Bielefeld,
            33615~Bielefeld,~Germany}\\
    {\small kondrat@mathematik.uni-bielefeld.de}
\and
    {\bf Oleksandr Kutoviy}\\
    {\small Department of Mathematics, MIT,}\\
    {\small 77~Massachusetts~Avenue~2-155, Cambridge, MA, USA;}\\
    {\small Fakult\"at f\"ur Mathematik, Universit\"at Bielefeld,
            33615~Bielefeld,~Germany}\\
            {\small kutovyi@mit.edu}
\and
    {\bf Maria Jo\~{a}o Oliveira}\\
    {\small Universidade Aberta, 1269-001 Lisbon, Portugal;}\\
    {\small CMAF, University of Lisbon, 1649-003 Lisbon, Portugal}\\
    {\small oliveira@cii.fc.ul.pt}}

\title{Dynamical Widom--Rowlinson model and~its~mesoscopic~limit}

\date{}
\maketitle

\begin{abstract}
We consider the non-equilibrium dynamics for the Widom--Rowlin\-son model
(without hard-core) in the continuum. The Lebowitz--Penrose-type scaling of
the dynamics is studied and the system of the corresponding kinetic equations
is derived. In the space-homogeneous case, the equilibrium points of this
system are described. Their structure corresponds to the dynamical phase
transition in the model. The bifurcation of the system is shown.
\end{abstract}

\newpage

\section{Introduction}

Critical behavior of complex systems in the continuum is one of the central
problems in statistical physics. For systems in $\R^d$, $d>1$, consisting of
particles of the same type there is, up to our knowledge, only one rigorous
mathematical analysis of this problem, the so-called LMP
(Lebowitz--Mazel--Presutti) models with Kac potentials, see
\cite[Chapter 10]{Pre2009} and the references therein. The case of particles
of different types has been more extensively studied. The simplest model was
proposed by Widom and Rowlinson \cite{WR1970} for a potential with hard-core.
In this model, there is an interaction only between particles of different
types. For large activity, the existence of phase transition for the model
in \cite{WR1970} was shown by Ruelle \cite{Rue1971}. A natural
modification of this model for the case of three or more different particle
types is the Potts model in the continuum. Within this context, Lebowitz and
Lieb \cite{LeLi1972} extended Ruelle's result to the multi-types case and
soft-core potentials. For a large class of potentials (with or without
soft-core), Georgii and H\"{a}ggstr\"{o}m \cite{GH1996} established the phase
transition. Further activity in this area concerns a mean-field theory for the
Potts model in the continuum and, in particular, for the Widom--Rowlinson
model, without hard-core, see \cite{GMRZ2006} for the most general case (that
is, two or more different types) and \cite{dMMPV2008,dMMPV2009} (for three or
more different types).

All these works deal with Gibbs equilibrium states of continuous particle
systems. Another approach to study Gibbs measures goes back to Glauber and
Dobrushin and it consists in the analysis of the stochastic dynamics
associated with these measures. In the continuous case, an analogue of the
Glauber dynamics is a spatial birth-and-death process whose intensities
imply the invariance of the dynamics with respect to a proper Gibbs
measure (the so-called detailed balance conditions). For continuous particle
systems of only one type, the corresponding non-equilibrium dynamics was
recently intensively studied, see e.g.~\cite{FKK2010,FKO2011a} and the
references therein. In this work we consider the corresponding Glauber-type
dynamics in the continuum, but for two different particle types. Here we use
the statistical Markov evolution rather than the dynamics in the sense of
trajectories. In other words, we study the dynamics in terms of states. This
can be done using the language of correlation functions corresponding to the
states or the language of the corresponding generating functionals.

We construct this dynamics for the Widom--Rowlinson model and study its
mesoscopic behavior under the so-called Lebowitz--Penrose scaling (see
\cite{Pre2009} and the references therein). For this purpose, we exploit a
technique based on the Ovsjannikov theorem, see e.g.~\cite{FKO2011a} and the
references therein. This allows us to derive rigorously the system of kinetic
equations for the dynamics, which critical behavior reflects the phase
transition phenomenon in the original microscopic dynamics. This scheme to
derive the kinetic equations for Markov evolutions in the continuum was
proposed in \cite{FKK2010a} and goes back to an approach well-known for the
Hamiltonian dynamics, see \cite{Spo1980}. Another approach is based on
minimizing some energy functionals, see e.g.~\cite{CCELM20099,CCELM2005}.

In Section~\ref{Section2} we briefly recall some notions of the analysis on
one- and two-types configuration spaces. A more
detailed explanation can be found in e.g.~\cite{AKR1998a,KK2002} and
\cite{Fin2009,FKO2011c}, respectively. We introduce and study a generalization
of generating functionals for two-types spaces as well. In
Section~\ref{Section3} we consider the dynamical Widom--Rowlinson model. We
prove that the corresponding time evolution in terms of entire generating
functionals exist in a scale of Banach spaces, for a finite time interval
(Theorem~\ref{Theorem1}). Section~\ref{Section4} is devoted to the mesoscopic
scaling in the Lebowitz--Penrose sense. We prove that the rescaled
evolution of entire generating functionals converges strongly to the limiting
time evolution (Theorem~\ref{Theorem2}). The latter preserves exponential
functionals (Theorem~\ref{reducedtokinetic}), which corresponds to the
propagation of the chaos principle for correlation functions,
cf.~e.g.~\cite{FKK2011a}. This allows to derive a system of kinetic
equations~\eqref{kinetic}, which are non-linear and non-local (they include
convolutions of functions on $\R^d$, cf.~e.g.~\cite{FKK2010b}). We also prove
the existence and uniqueness of the solutions to the aforementioned system of
equations (Theorem~\ref{Theorem3}). In Section~\ref{Section5} we consider
the same system but in the space-homogeneous case. Even this simplest case
reflects the dynamical phase transition, which is expected to occur in the
original non-equilibrium dynamics. Namely, in Theorem~\ref{th-homog-kinetic}
we prove that there is a critical value below of which the system of kinetic
equations will have a unique stable equilibrium point. For higher
values, the system will have three equilibrium points: two of them are
stable (they correspond to the pure phases of the reversible Gibbs measure)
and the third one is unstable (it corresponds to the symmetric mixed phase).
It is worth noting that at the critical point the system has a unique mixed
equilibrium point (saddle-node). Therefore, there is a bifurcation in the
system of the kinetic equations corresponding to the dynamical
Widom--Rowlinson model without hard-core. We note that the existence of such a
critical value for the original (equilibrium) model was an open problem, see
e.g.~\cite[Remark~1.3]{GH1996}.

\section{General Framework}\label{Section2}

This section begins by briefly recalling the concepts and results of
combinatorial harmonic analysis on one- and two-types configuration spaces
needed throughout this work. For a detailed explanation see
e.g.~\cite{AKR1998a,Fin2009,FKO2011c,KK2002} and the references cited therein.

\subsection{One-component configuration spaces}\label{Subsection21}

The configuration space $\Gamma:=\Gamma_{\R}$ over $\R^d$, $d\in\N$, is
defined as the set of all locally finite subsets (configurations) of $\R^d$,
\[
\Gamma:=\left\{\gamma\subset\R^d : |\gamma\cap\Lambda| <\infty \hbox{ for
every compact } \Lambda\subset \R^d \right\} ,
\]
where $\left| \cdot \right|$ denotes the cardinality of a set. We will identify a configuration $\gamma \in \Gamma $ with the non-negative Radon measure
$\sum_{x\in \gamma }\delta_x$ on the Borel $\sigma$-algebra
$\mathcal{B}(\mathbb{R}^d)$, where $\delta_x$ is the Dirac measure with mass $1$
at $x$ and $\sum_{x\in \emptyset}\delta_x:=0$. This identification allows to
endow $\Gamma$ with the vague topology and the corresponding Borel $\sigma$-algebra
$\mathcal{B}(\Gamma)$.

Let $\rho>0$ be a locally integrable function on $\R^d$. The Poisson measure
$\pi_\sigma$ with intensity the Radon measure
\[
d\sigma(x)=\rho(x)dx
\]
is defined as the probability measure
on $(\Gamma, \mathcal{B}(\Gamma))$ with Laplace transform given by
\[
\int_\Gamma d\pi_\sigma(\gamma)\,\exp \biggl( \sum_{x\in \gamma }\varphi (x)\biggr)
=\exp \biggl(\int_{\R^d}dx\,\rho(x) \left( e^{\varphi (x)}-1\right)\biggr)
\]
for all smooth functions $\varphi$ on $\R^d$ with compact support. For the case $\rho\equiv1$, we will omit the index $\pi:=\pi_{dx}$.

For any $n\in\N_0:=\N\cup\{0\}$, let
\[
\Gamma^{(n)}:= \{ \gamma\in \Gamma: \vert \gamma\vert = n\},\ n\in \N,\quad \Gamma^{(0)} := \{\emptyset\}.
\]
Clearly, each $\Gamma^{(n)}$, $n\in\N$, can be identified with the
symmetrization of the set $ \{(x_1,...,x_n)\in (\R^d)^n: x_i\not=
x_j \hbox{ if } i\not= j\}$ under the permutation group over
$\{1,...,n\}$, which induces a natural (metrizable) topology on
$\Gamma^{(n)}$ and the corresponding Borel $\sigma$-algebra
$\mathcal{B}(\Gamma^{(n)})$. Moreover, for the product measure
$\sigma^{\otimes n}$ fixed on $(\R^d)^n$, this identification yields a
measure $\sigma^{(n)}$ on $(\Gamma^{(n)},\mathcal{B}(\Gamma^{(n)}))$. This leads to
the space of finite configurations
\[
\Gamma_0 := \bigsqcup_{n=0}^\infty \Gamma^{(n)}
\]
endowed with the topology of disjoint union of topological spaces and the
corresponding Borel $\sigma$-algebra $\mathcal{B}(\Gamma_0)$, and to the
so-called Lebesgue--Poisson measure on $(\Gamma_0,\mathcal{B}(\Gamma_0))$,
\[
\lambda_\sigma:=\sum_{n=0}^\infty \frac{1}{n!} \sigma^{(n)},\quad \sigma^{(0)}(\{\emptyset\}):=1.
\]
We set $\lambda:=\lambda_{dx}$.

\subsection{Two-component configuration spaces}\label{Subsection22}

The previous definitions can be naturally extended to $n$-component configuration
spaces. Having in mind our goals, we just present the extension for $n=2$.

Given two copies of the space $\Gamma$, denoted by $\Gamma^+$ and $\Gamma^-$,
let
\[
\Gamma^2:=\left\{(\gamma^+,\gamma^-)\in\Gamma^+\times\Gamma^-:
\gamma^+\cap\gamma^-=\emptyset\right\}.
\]

Similarly, given two copies of the space $\Gamma_0$, denoted by $\Gamma_0^+$ and
$\Gamma_0^-$, we consider the space
\[
\Gamma^2_0:=\left\{(\eta^+,\eta^-)\in\Gamma_0^+\times\Gamma_0^-:
\eta^+\cap\eta^-=\emptyset\right\}.
\]

We endow $\Gamma^2$ and $\Gamma_0^2$ with the topology induced by the
product of the topological spaces $\Gamma^+\times\Gamma^-$ and
$\Gamma_0^+\times\Gamma_0^-$, respectively, and with the corresponding
Borel $\sigma$-algebras, denoted by $\mathcal{B}(\Gamma^2)$ and
$\mathcal{B}(\Gamma^2_0)$.

We consider the space
$B_\mathrm{ls}(\Gamma_0^2)$ of all complex-valued
$\mathcal{B}(\Gamma^2_0)$-measurable functions $G$ with local support, i.e.,
$G\!\!\upharpoonright _{\Gamma_0^2\backslash\left(\Gamma_{\Lambda}^+\times\Gamma_{\Lambda}^-\right)}\equiv 0$ for some bounded set
$\Lambda\in\mathcal{B}(\R^d)$, where
$\Gamma_{\Lambda}^\pm:=\{\eta\in\Gamma^\pm:\eta\subset\Lambda\}$. Given a
$G\in B_\mathrm{ls}(\Gamma_0^2)$, the $\KK$-transform of $G$ is the mapping
$\KK G:\Gamma^2\to\C$ defined at each $(\gamma^+,\gamma^-)\in\Gamma^2$ by
\begin{equation}
(\KK G)(\gamma^+,\gamma^-):= \sum_{\substack{\eta^+\subset\gamma^+ \\
\vert\eta^+\vert < \infty}} \sum_{\substack{\eta^-\subset\gamma^-\\
\vert\eta^-\vert < \infty}} G(\eta^+,\eta^-). \label{M1}
\end{equation}
Note that for every $G\in B_\mathrm{ls}(\Gamma_0^2)$ the sum in \eqref{M1} has
only a finite number of summands different from zero and thus $\KK G$ is a
well-defined function on $\Gamma^2$. Moreover,
$\KK :B_\mathrm{ls}(\Gamma_0^2)\to\KK(B_\mathrm{ls}(\Gamma_0^2))$ is a
positivity preserving linear isomorphism whose inverse mapping is defined by
\[
\left(\KK ^{-1}F\right) (\eta^+,\eta^-):=
\sum_{\xi^+\subset\eta^+}\sum_{\xi^-\subset\eta^-}
(-1)^{|\eta^+\backslash\xi^+|+|\eta^-\backslash\xi^-|}F(\xi^+,\xi^-),\quad (\eta^+,\eta^-)\in\Gamma_0^2.
\]

The $\KK$-transform might be extended pointwisely to a wider class of functions. Among them we will distinguish
the so-called Lebesgue--Poisson exponentials $e_\lambda(f^+,f^-)$ defined for
complex-valued $\mathcal{B}(\R^d)$-meas\-urable functions $f^+,f^-$ by
\begin{equation}\label{ela2}
e_\lambda (f^+,f^-;\eta^+,\eta^-):=e_\lambda (f^+,\eta^+)e_\lambda (f^-,\eta^-),\quad (\eta^+,\eta^-)\in\Gamma_0^2,
\end{equation}
where
\begin{equation*}
e_\lambda (f^\pm,\eta^\pm):=\prod_{x\in \eta^\pm}f^\pm(x),\ \eta^\pm \in
\Gamma^\pm_0\!\setminus\!\{\emptyset\},\quad  e_\lambda(f^\pm,\emptyset ):=1.
\end{equation*}
Indeed, for any $f^+,f^-$ described as before, having in addition compact support, for
all $(\gamma^+,\gamma^-)\in \Gamma^2$
\begin{equation}
\left(\KK e_\lambda (f^+,f^-)\right) (\gamma^+,\gamma^-)=
\prod_{x\in \gamma^+}\bigl(1+f^+(x)\bigr)\prod_{y\in \gamma^-}\bigl(1+f^-(y)\bigr).\label{Kcog}
\end{equation}
The special role of functions \eqref{ela2} is partially due to the fact that the right-hand side of \eqref{Kcog} coincides with the integrand functions of
generating functionals (Subsection \ref{Subsection23} below).

Let now $\mathcal{M}_{\mathrm{fm}}^1(\Gamma^2)$ be the set of all
probability measures $\mu$ on $(\Gamma^2,\mathcal{B}(\Gamma^2))$ with
finite local moments of all orders, i.e., for all $n\in\N$ and all
bounded sets $\Lambda\in\mathcal{B}(\R^d)$
\[
\int_{\Gamma^2} d\mu(\gamma^+,\gamma^-)\, |\gamma^+\cap\Lambda|^n|\gamma^-\cap\Lambda|^n<\infty,
\]
and let $B_{\mathrm{bs}}(\Gamma_0^2)$ be the set of all bounded functions
$G\in B_\mathrm{ls}(\Gamma_0^2)$ such that $G\!\!\upharpoonright _{\Gamma_0^2\backslash\left(\bigsqcup_{n=0}^{N^+}\Gamma_{\Lambda^+}^{(n)}\times\bigsqcup_{n=0}^{N^-}\Gamma_{\Lambda^-}^{(n)}\right)}\equiv 0$ for some $N^+,N^-\in\N_0$ and for some
bounded Borel sets $\Lambda^+,\Lambda^-\subset\R^d$. Here, for $k\in\N_0$ and
for bounded sets $\Lambda^\pm\in\mathcal{B}(\R^d)$,
$\Gamma_{\Lambda^\pm}^{(k)}:=\{\eta\in\Gamma^\pm_{\Lambda^\pm}:|\eta|=k\}$. Given a
$\mu\in\mathcal{M}_{\mathrm{fm}}^1(\Gamma^2)$, the so-called
correlation measure $\rho_\mu$ corresponding to $\mu$ is a measure
on $(\Gamma_0^2,\mathcal{B}(\Gamma_0^2))$ defined for all $G\in
B_{\mathrm{bs}}(\Gamma_0^2)$ by
\begin{equation}
\int_{\Gamma_0^2}d\rho_\mu(\eta^+,\eta^-)\,G(\eta^+,\eta^-)=\int_{\Gamma^2} d\mu(\gamma^+,\gamma^-)\,
\left(\KK G\right) (\gamma^+,\gamma^-).  \label{Eq2.16}
\end{equation}
Note that under these assumptions $\KK \left|G\right|$ is $\mu$-integrable,
and thus \eqref{Eq2.16} is well-defined. In terms of correlation measures, this
shows, in particular, that
$B_{\mathrm{bs}}(\Gamma_0^2)\subset L^1(\Gamma_0^2,\rho_\mu)$.\footnote{Throughout this work all $L^p$-spaces, $p\geq 1$, consist of complex-valued functions.}
Actually, $B_{\mathrm{bs}}(\Gamma_0^2)$ is dense in $L^1(\Gamma_0^2,\rho_\mu)$.
Moreover, still by \eqref{Eq2.16}, on $B_{\mathrm{bs}}(\Gamma_0^2)$ the
inequality
$\Vert \KK G\Vert_{L^1(\Gamma^2,\mu)}\leq \Vert G\Vert_{L^1(\Gamma_0^2,\rho_\mu)}$
holds, allowing an extension of the $\KK$-transform to a bounded operator
$\KK:L^1(\Gamma_0^2,\rho_\mu)\to L^1(\Gamma^2,\mu)$ in such a way that
equality \eqref{Eq2.16} still holds for any $G\in
L^1(\Gamma_0^2,\rho_\mu)$. For the extended operator, the explicit form
\eqref{M1} still holds, now $\mu$-a.e. In particular, for functions
$f^+,f^-$ such that $e_\lambda (f^+,f^-)\in L^1(\Gamma_0^2,\rho_\mu)$ equality
\eqref{Kcog} still holds, but only for
$\mu$-a.a.~$(\gamma^+,\gamma^-)\in\Gamma^2$.

Let us now consider two measures on $\R^d$, in general different,
$d\sigma^\pm=\rho^\pm dx$, both defined as above.
The Lebesgue--Poisson product measure $\lambda^2_{\sigma^+,\sigma^-}:=\lambda_{\sigma^+}\otimes\lambda_{\sigma^-}$ on
$(\Gamma_0^2,\mathcal{B}(\Gamma_0^2))$ is the correlation measure
corresponding to the Poisson product measure $\pi^2_{\sigma^+,\sigma^-}:=\pi_{\sigma^+}\otimes\pi_{\sigma^-}$ on
$(\Gamma^2,\mathcal{B}(\Gamma^2))$. Observe that a priori $\lambda^2_{\sigma^+,\sigma^-}$ is a measure defined on
$(\Gamma_0\times\Gamma_0,\mathcal{B}(\Gamma_0)\otimes\mathcal{B}(\Gamma_0))$
and $\pi^2_{\sigma^+,\sigma^-}$ is a measure defined on $(\Gamma\times\Gamma,\mathcal{B}(\Gamma)\otimes\mathcal{B}(\Gamma))$. It can
actually be shown that $\Gamma_0^2=(\Gamma_0\times\Gamma_0)\setminus\{(\eta,\xi): \eta\cap\xi\not=\emptyset\}$ has full $\lambda^2_{\sigma^+,\sigma^-}$-measure and
$\Gamma^2=(\Gamma\times\Gamma)\setminus\{(\gamma,\gamma'): \gamma\cap\gamma'\not=\emptyset\}$ has full $\pi^2_{\sigma^+,\sigma^-}$-measure, cf.~\cite{Fin2009,Kun1999}.

It can also be shown that $e_\lambda(f^+,f^-)\in L^p(\Gamma_0^2, \lambda^2_{\sigma^+,\sigma^-})$ whenever $f^\pm\in L^p(\R^d,\sigma^\pm)$ for some $p\geq 1$, and,
moreover,
\[
\Vert e_\lambda(f^+,f^-)\Vert^p_{L^p(\Gamma_0^2,\lambda^2_{\sigma^+,\sigma^-})}
=\exp(\Vert f^+\Vert^p_{L^p(\R^d,\sigma^+)}+\Vert f^-\Vert^p_{L^p(\R^d,\sigma^-)}).
\]
In particular, for $p=1$, one additionally has, for all
$f^\pm\in L^1(\R^d,\sigma^\pm)$,
\begin{multline}
\int_{\Gamma_0^2}d\lambda^2_{\sigma^+,\sigma^-}(\eta^+,\eta^-)\, e_\lambda(f^+,f^-;\eta^+,\eta^-)
\\= \exp\biggl(\int_{\R^d}dx\left(\rho^+(x)f^+(x)+\rho^-(x)f^-(x)\right)\biggr).\label{meanLP}
\end{multline}

In the sequel we set
\begin{equation*}
  \lambda^2:=\lambda^2_{dx,dx}, \qquad \pi^2:=\pi^2_{dx,dx}.
\end{equation*}

\subsection{Bogoliubov generating functionals}\label{Subsection23}

The notion of Bogoliubov generating functional corresponding to a probability
measure on $(\Gamma,\mathcal{B}(\Gamma))$ \cite{KKO2006} naturally extends
to probability measures defined on a multicomponent space
$(\Gamma^n,\mathcal{B}(\Gamma^n))$. For simplicity, we just
present the extension for $n=2$. Of course, a similar procedure is used for
$n>2$, but with a more cumbersome notation.

\begin{definition}\label{Definition1}
Given a probability measure $\mu$ on $(\Gamma^2, \mathcal{B} (\Gamma^2))$, the
Bogoliubov generating functional (shortly GF) $B_\mu$ corresponding to $\mu$
is the functional defined at each pair $\theta^+$, $\theta^-$ of
complex-valued $\mathcal{B}(\R^d)$-measurable functions by
\begin{equation*}
B_\mu(\theta^+,\theta^-) :=\int_{\Gamma^2} d\mu(\gamma^+,\gamma^-)\,
\prod_{x\in\gamma^+}\bigl(1+\theta^+(x)\bigr)\prod_{y\in\gamma^-}\bigl(1+\theta^-(y)\bigr),
\end{equation*}
provided the right-hand side exists.
\end{definition}

Clearly, for an arbitrary probability measure $\mu$, $B_\mu$ is always defined at least at the pair $(0,0)$. However, the whole domain of $B_\mu$ depends on properties of the underlying measure $\mu$. For
instance, probability measures $\mu$ for which the GF is well-defined on
multiples of indicator functions $\1_\Lambda$ of bounded Borel sets
$\Lambda$, necessarily have finite local exponential moments, i.e.,
\begin{equation}\label{expint}
\int_{\Gamma^2} d\mu (\gamma^+,\gamma^-)\, e^{\alpha(|\gamma^+\cap\Lambda |+|\gamma^-\cap\Lambda |)}<\infty,\quad \mathrm{for\,\,all}\,\,\alpha>0.
\end{equation}
The converse is also true. In fact, for all $\alpha>0$ and for all
$\Lambda$ described as before we have that the left-hand side of \eqref{expint} is equal to
\[
\int_{\Gamma^2}d\mu (\gamma^+,\gamma^-)\,\prod_{x\in \gamma^+\cup\gamma^-}e^{\alpha\1_\Lambda (x)}\\
=B_\mu\bigl((e^\alpha-1)\1_\Lambda,(e^\alpha-1)\1_\Lambda\bigr) < \infty.
\]
According to the previous subsection, this implies that to such a measure
$\mu$ one may associate the correlation measure $\rho_\mu$, leading to a
description of the functional $B_\mu$ in terms of the measure $\rho_\mu$:
\begin{align*}
B_\mu(\theta^+,\theta^-)
&= \int_{\Gamma^2} d\mu(\gamma^+,\gamma^-)\,\left(\KK e_\lambda (\theta^+,\theta^-)\right) (\gamma^+\,\gamma^-)\\
&= \int_{\Gamma_0^2}d\rho_\mu(\eta^+,\eta^-)\, e_\lambda(\theta^+,\theta^-;\eta^+,\eta^-),
\end{align*}
or in terms of the so-called correlation function
\[
k_\mu:=\dfrac{d\rho_\mu}{d\lambda^2}
\]
corresponding to the measure $\mu$, provided $\rho_\mu$ is absolutely continuous
with respect to the product measure $\lambda^2$:
\begin{equation}
B_\mu(\theta^+,\theta^-)=\int_{\Gamma_0^2}d\lambda^2(\eta^+,\eta^-)\,
e_\lambda (\theta^+,\theta^-;\eta^+,\eta^-)k_\mu(\eta^+,\eta^-).\label{BF_via_cf}
\end{equation}

Throughout this work we will consider GF which are entire on the whole
$L^1(\R^d,dx)\times L^1(\R^d,dx)$ space with
the norm
\[
\|(\theta^+,\theta^-)\|_{L^1\times L^1}= \llvert\theta^+\rrvert_1+\llvert\theta^-\rrvert_1.
\]
Here and below we use the notation
\begin{equation*}
 \llvert\theta\rrvert_1:=\|\theta\|_{L^1} , \qquad \theta\in L^1 := L^1(\R^d,dx).
\end{equation*}
We recall that a functional $A:L^1\times L^1\to\C$ is entire on
$L^1\times L^1$ whenever $A$ is locally bounded and for all
$\theta_0^\pm,\theta^\pm\in L^1$ the mapping
\[
\C^2\ni (z^+,z^-)\mapsto A(\theta_0^++ z^+\theta^+,\theta_0^-+ z^-\theta^-)\in\C
\]
is entire \cite{Kr82}, which is equivalent to entireness on $L^1$ of $A$ on each component.
Thus, at each pair $\theta_0^+,\theta_0^-\in L^1$, every entire functional $A$
on $L^1\times L^1$ has a representation in terms of its Taylor expansion,
\begin{multline*}
A(\theta_0^++ z^+\theta^+,\theta_0^-+ z^-\theta^-)\\=\sum_{n,m=0}^\infty \frac{(z^+)^n(z^-)^m}{n!m!}
d^{(n,m)}A(\theta_0^+,\theta_0^-;\underbrace{\theta^+,\ldots,\theta^+}_{n\ \text{times}},\underbrace{\theta^-,\ldots,\theta^-}_{m\ \text{times}}),
\end{multline*}
$z^\pm\in\C$, $\theta^\pm\in L^1$. Extending the kernel theorem
\cite[Theorem 5]{KKO2006} to the two-component case, each
differential $d^{(n,m)}A(\theta_0^+,\theta_0^-;\cdot)$ is then defined by a kernel
$\vd^{(n,m)}A(\theta_0^+,\theta_0^-;\cdot)\in L^\infty((\R^d)^n\times(\R^d)^m)$,
which is symmetric in the first $n$ coordinates and in the last $m$
coordinates. More precisely,
\begin{align}
&d^{(n,m)}A(\theta _0^+,\theta_0^-;\theta _1^+,\ldots,\theta_n^+,\theta^-_1,\ldots,\theta^-_m)\nonumber\\
&=\frac{\partial
^{n+m}}{\partial
z_1^+...\partial z_n^+\partial z_1^-...\partial z_m^-}
A\biggl( \theta _0^++\sum_{i=1}^nz_i^+\theta _i^+,\theta _0^-+\sum_{j=1}^mz_j^-\theta _j^- \biggr)
\biggr\vert_{z_1^+=...=z_n^+=z_1^-=...=z^-_ m=0} \nonumber\\
&=\int_{(\R^d)^n\times(\R^d)^m}dx_1\ldots dx_ndy_1\ldots dy_m\,\nonumber\\
&\qquad\qquad\times\vd^{(n,m)}A(\theta_0^+,\theta_0^-;x_1,\ldots,x_n,y_1,\ldots,y_m)\prod_{i=1}^n\theta_i^+(x_i)
\prod_{j=1}^m\theta_j^-(y_j),\nonumber
\end{align}
for all $\theta _1^+,...,\theta _n^+,\theta^-_1,...,\theta^-_m\in L^1$, $n,m\in\N$.
Moreover, the operator norm of the bounded $(n+m)$-linear functional
(on $L^1\times L^1$) $d^{(n,m)}A(\theta_0^+,\theta_0^-;\cdot)$ is equal to
$\left\| \vd^{(n,m)}A(\theta_0^+,\theta_0^-;\cdot)\right\|_{L^\infty((\R^d)^n\times(\R^d)^m)}$
and for all $r>0$ one has
\begin{multline*}
\left\|\vd^{(n,m)}{A}(\theta_0^+,\theta_0^-;\cdot)\right\|_{L^\infty((\R^d)^n\times(\R^d)^m)}\\
\leq  n!m!
\left(\frac{e}{r}\right)^{n+m} \sup_{\llvert\theta^\pm\rrvert_1\leq r}
|A(\theta_0^++\theta^+,\theta_0^-+\theta^-)|.
\end{multline*}

For the cases where either $n=0$ or $m=0$, the entireness property on $L^1$ of each
pair of functionals $A(\cdot,\theta_0^-)$, $A(\theta_0^+,\cdot)$ implies
by a direct application of \cite[Theorem 5]{KKO2006} that the corresponding
differentials are defined by a symmetric kernel
$\vd^n{A}(\theta_0^+,\theta_0^-;\cdot,\emptyset)\in L^\infty((\R^d)^n)$,
$\vd^m{A}(\theta_0^+,\theta_0^-;\emptyset,\cdot)\in L^\infty((\R^d)^m)$, respectively,
and for each $r>0$ one has
\begin{align*}
&\left\| \vd A(\theta_0^+,\theta_0^-;\cdot,\emptyset)\right\|_{L^\infty} \leq \frac{1}{r} \sup_{\llvert\theta^+\rrvert_1 \leq r}
|A(\theta_0^++\theta^+,\theta_0^-)|,\\
&\left\| \vd A(\theta_0^+,\theta_0^-;\emptyset,\cdot)\right\|_{L^\infty} \leq \frac{1}{r} \sup_{\llvert\theta^-\rrvert_1 \leq r}
|A(\theta_0^+,\theta_0^-+\theta^-)|
\end{align*}
and, for $n,m\geq 2$,
\begin{align*}
\left\|\vd^n{A}(\theta_0^+,\theta_0^-;\cdot,\emptyset)\right\|_{L^\infty((\R^d)^n)} &\leq n!
\left(\frac{e}{r}\right)^n \sup_{\llvert\theta^+\rrvert_1 \leq r}
|A(\theta_0^++\theta^+,\theta_0^+)|,\\
\left\|\vd^m{A}(\theta_0^+,\theta_0^-;\emptyset,\cdot)\right\|_{L^\infty((\R^d)^m)} &\leq m!
\left(\frac{e}{r}\right)^m \sup_{\llvert\theta^-\rrvert_1 \leq r}
|A(\theta_0^+,\theta_0^-+\theta^-)|.
\end{align*}

All these considerations hold, in particular, for $A$ being an entire GF $B_\mu$
on $L^1\times L^1$ corresponding to some probability measure $\mu$ on
$\Gamma^2$ which is locally absolutely continuous with respect to $\pi^2$,
that is, for all disjoint bounded Borel sets $\Lambda^+,\Lambda^-\subset\R^d$
the image measure $\mu\circ p_{\Lambda^+,\Lambda^-}^{-1}$ of $\mu$ under the
projection $p_{\Lambda^+,\Lambda^-}(\gamma^+,\gamma^-):=(\gamma^+\cap\Lambda^+,\gamma^-\cap\Lambda^-)\in\Gamma_{\Lambda^+}\times\Gamma_{\Lambda^-}$ is absolutely
continuous with respect to the product of image measures
$(\pi\circ p_{\Lambda^+}^{-1})\times(\pi\circ p_{\Lambda^-}^{-1})$,
$p_{\Lambda^\pm}(\gamma^\pm):=\gamma^\pm\cap\Lambda^\pm$. In this case, the
correlation function $k_\mu$ exists and it is given for
$\lambda^2$-a.a.~$(\eta^+,\eta^-)\in\Gamma_0^2$ by
\begin{align*}
k_\mu(\eta^+,\eta^-)&=\vd^{(n,m)} B_\mu(0,0;\eta^+,\eta^-),\\
k_\mu(\eta^+,\emptyset)&=\vd^n B_\mu(0,0;\eta^+,\emptyset),\\
k_\mu(\emptyset,\eta^-)&=\vd^m B_\mu(0,0;\emptyset,\eta^-),
\end{align*}
for $|\eta^+|=n$, $|\eta^-|=m$, $n,m\in\N$ \cite[Proposition 9]{KKO2006}.
As a consequence, similarly to \cite[Proposition 11]{KKO2006} one finds
\begin{multline}
\delta^{(|\eta^+|,|\eta^-|)}B_\mu(\theta^+,\theta^-;\eta^+,\eta^-)\\=
\int_{\Gamma _0^2}d\lambda^2(\xi^+,\xi^-)\,k_\mu(\xi^+\cup \eta^+,\xi^-\cup\eta^-)
e_\lambda(\theta^+,\theta^-;\xi^+,\xi^-)\label{Equation3}
\end{multline}
for $\lambda^2$-a.a.~$(\eta^+,\eta^-)\in\Gamma_0^2$, which gives an
alternative description of the kernels
$\vd^{(n,m)}B_\mu(\theta_0^+,\theta_0^-;\cdot)$, $n,m\in\N_0$. Moreover,
the previous estimates for the norms of the kernels lead to the so-called
Ruelle generalized bound \cite[Proposition 16]{KKO2006}, that is, for any
$0\leq\eps\leq 1$ and any $r>0$ there is a constant $C\geq 0$ depending on $r$
such that
\[
k_\mu(\eta^+,\eta^-)\leq C\left(\lvert\eta^+\rvert!\lvert\eta^-\rvert!\right)^{1-\eps}
\left( \frac er\right) ^{\lvert\eta^+\rvert+\lvert\eta^-\rvert},\quad
\lambda^2 \mathrm{-a.a.\ }(\eta^+,\eta^-)\in\Gamma_0^2.
\]
Similarly to the one-component case \cite[Proposition 23]{KKO2006}, the latter
motivates for each $\alpha>0$ the definition of the Banach space $\mathcal{E}_\alpha$ of
all entire functionals $B$ on $L^1\times L^1$ such that
\begin{equation}\label{norminscale}
\left\| B\right\| _\alpha :=\sup_{\theta^+,\theta^- \in L^1}
\left( \left|B(\theta^+,\theta^-)\right| e^{-\frac{1}{\alpha} \left(\llvert \theta^+ \rrvert _1 +\llvert \theta^- \rrvert _1 \right)}\right)<\infty.
\end{equation}
Observe that this class of Banach spaces has the property that,
for each $\alpha_0>0$, the family
$\{\mathcal{E}_\alpha: 0<\alpha\leq\alpha_0\}$ is a scale of Banach spaces,
that is,
\[
\mathcal{E}_{\alpha''}\subseteq \mathcal{E}_{\alpha'},\quad
\|\cdot\|_{\alpha'}\leq \|\cdot\|_{\alpha''}
\]
for any pair $\alpha'$, $\alpha''$ such that $0<\alpha'< \alpha''\leq\alpha_0$.

Of course these considerations hold, more generally, for any entire
functional $B$ on $L^1\times L^1$ of the form
\[
B(\theta^+,\theta^-)=\int_{\Gamma_0^2}d\lambda^2(\eta^+,\eta^-)\,
e_\lambda(\theta^+,\theta^-;\eta^+,\eta^-)k(\eta^+,\eta^-),\quad k:\Gamma_0^2\to\left[0,+\infty\right).
\]

\section{The Widom--Rowlinson model}\label{Section3}

The dynamical Widom--Rowlinson model is an example of a birth-and-death model
of two different particle types, let us say $+$ and $-$, where, at each random
moment of time, $+$ and $-$ particles randomly disappear according to a death
rate identically equal to $m>0$, while new $\pm$ particles
randomly appear according to a birth rate which only depends on the
configuration of the whole $\mp$-system at that time. The influence of the
$\pm$-system of particles is none in this process. More precisely,
let $\phi:\R^d\to\R\cup\{+\infty\}$ be a pair potential, that is, a
$\mathcal{B}(\R^d)$-measurable function such that $\phi(-x)=\phi(x)\in \R$ for
all $x\in \R^d\setminus\{0\}$, which we will assume to be non-negative and
integrable. Given a configuration $(\gamma^+,\gamma^-)\in\Gamma^2$, the birth
rate of a new $+$ particle at a site $x\in\R^d\setminus(\gamma^+\cup\gamma^-)$
is given by $\exp(-E(x,\gamma^-))$, where $E(x,\gamma^-)$ is a relative energy
of interaction between a particle located at $x$ and the configuration
$\gamma^-$ defined by
\[
E(x,\gamma^-):=\sum_{y\in \gamma^-}\phi (x-y)\in\left[0,+\infty\right].
\]
Similarly, the birth rate of a new $-$ particle at a site
$y\in\R^d\setminus(\gamma^+\cup\gamma^-)$ is given by $\exp(-E(y,\gamma^+))$.

Informally, the behavior of such an infinite particle system is described by a pre-generator\footnote{Here and below, for simplicity of notation, we have just written $x$, $y$
instead of $\{x\}$, $\{y\}$, respectively.}
\begin{align*}
(LF)(\gamma^+,\gamma^-):=&\,m\sum_{x\in \gamma^+}\left(F(\gamma^+\setminus x,\gamma^-) -
F(\gamma^+,\gamma^-)\right)\\
&+m\sum_{y\in \gamma^-}\left(F(\gamma^+,\gamma^-\setminus y) - F(\gamma^+,\gamma^-)\right)\\
&+ z\int_{\R^d} dx\,e^{-E(x,\gamma^-)} \left(F(\gamma^+\cup x,\gamma^-) - F(\gamma^+,\gamma^-)\right)\\
&+z\int_{\R^d} dy\,e^{-E(y,\gamma^+)} \left(F(\gamma^+,\gamma^-\cup y) - F(\gamma^+,\gamma^-)\right),
\end{align*}
where $z>0$ is an activity parameter. Of course, the previous expression will
be well-defined under proper conditions on the function $F$ \cite{FKO2011c}.

In applications, properties of the time evolution of an infinite particle
system, like the described one, in terms of states, that is, probability
measures on $\Gamma^2$, are a subject of interest. Informally, such a time
evolution is given by the so-called Fokker--Planck equation
\begin{equation}
\frac{d\mu_t}{dt}=L^*\mu_t, \quad
{\mu_t}_{|t=0}=\mu_0\label{FokkerPlanck},
\end{equation}
where $L^*$ is the dual operator of $L$. As explained in \cite{FKO2011c},
technically the use of definition \eqref{Eq2.16} allows an alternative
approach to the study of \eqref{FokkerPlanck} through the corresponding
correlation functions $k_t:=k_{\mu_t}$, $t\geq0$, provided they exist. This
leads to the Cauchy problem
\begin{equation*}
\frac \partial {\partial t}k_t=\widehat L^*k_t,\quad
{k_t}_{|t=0}=k_0,
\end{equation*}
where $k_0$ is the correlation function corresponding to the initial
distribution $\mu_0$ of the system and $\widehat L^*$ is the dual operator of
$\widehat L:=\KK^{-1}L\KK$ in the sense
\begin{multline}
\int_{\Gamma_0^2}d\lambda^2(\eta^+,\eta^-)\,(\widehat LG)(\eta^+,\eta^-)
k(\eta^+,\eta^-)\\=\int_{\Gamma_0^2}d\lambda^2(\eta^+,\eta^-)\,G(\eta^+,\eta^-)
(\widehat L^*k)(\eta^+,\eta^-).\label{duality}
\end{multline}
To define $\widehat{L}$ and $\widehat{L}^*$ with a full rigor, see
\cite{FKO2011c}.

Now we would like to rewrite the dynamics \eqref{FokkerPlanck} in terms of the
GF $B_t$ corresponding to $\mu_t$. Through the representation
\eqref{BF_via_cf}, this can be done, informally, by
\begin{align}
\label{obtevol}\frac \partial {\partial t}B_t(\theta^+,\theta^-) &=\int_{\Gamma _0^2}d\lambda^2(\eta^+\eta^-)\,e_\lambda(\theta^+,\theta^-;\eta^+,\eta^-)\Bigl( \frac \partial{\partial t}k _t(\eta^+,\eta^-)\Bigr)\\
&=\int_{\Gamma _0^2}d\lambda^2(\eta^+,\eta^-)\, e_\lambda(\theta^+,\theta^-;\eta^+,\eta^-)(\widehat L^*k_t)(\eta^+,\eta^-)\nonumber\\
&=\int_{\Gamma _0^2}d\lambda^2(\eta^+,\eta^-)\,(\widehat
Le_\lambda(\theta^+,\theta^-))(\eta^+,\eta^-)k_t(\eta^+,\eta^-),\nonumber
\end{align}
for all $\theta^\pm\in L^1$. In other words, given the operator $\widetilde L$ defined at
\[
B(\theta^+,\theta^-):=\int_{\Gamma_0^2}d\lambda^2(\eta^+,\eta^-)\,
e_\lambda(\theta^+,\theta^-;\eta^+,\eta^-)k(\eta^+,\eta^-),
\]
for some $k:\Gamma_0^2\to\left[0,+\infty\right)$,
by
\begin{equation*}
(\widetilde LB)(\theta^+,\theta^-):=\int_{\Gamma_0^2}d\lambda^2(\eta^+,\eta^-)\,(\widehat Le_\lambda(\theta^+,\theta^-))(\eta^+,\eta^-)k(\eta^+,\eta^-),
\end{equation*}
one has that the GF $B_t$, $t\geq 0$ are a (pointwise) solution to the equation
\begin{equation*}
\frac{\partial B_t}{\partial t}=\widetilde LB_t.
\end{equation*}

We will find now an explicit expression for the operator $\widetilde{L}$. The
fact that the expression is well-defined will follow from Proposition
\ref{Proposition2} below.

\begin{proposition}\label{Proposition1}
For all $\theta^\pm\in L^1$, we have
\begin{align*}
&(\widetilde LB)(\theta^+,\theta^-)\\
=&-\int_{\R^d}dx\,\theta^+(x)\bigl(m\delta B(\theta^+,\theta^-; x,\emptyset)
-zB(\theta^+,\theta^- e^{-\phi (x -\cdot )}+e^{-\phi (x-\cdot )}-1)\bigr)\\
&-\int_{\R^d}dy\,\theta^-(y)\bigl(m\delta B(\theta^+,\theta^-;\emptyset, y)
-zB(\theta^+ e^{-\phi (y -\cdot )}+e^{-\phi (y-\cdot )}-1,\theta^-)\bigr).
\end{align*}
\end{proposition}

\begin{proof}
As shown in \cite[Sections 3 and 5]{FKO2011c},
\begin{align*}
&(\widehat LG)(\eta^+,\eta^-)\nonumber\\
=&-m\left(|\eta^+|+|\eta^-|\right)G(\eta^+,\eta^-)\nonumber\\
&+z\sum_{\xi^-\subseteq\eta^-}\int_{\R^d}dx\,G(\eta^+\cup x,\xi^-)e^{-E(x,\xi^-)}e_\lambda(e^{-\phi(x-\cdot)}-1,\eta^-\setminus\xi^-)\nonumber\\
&+z\sum_{\xi^+\subseteq\eta^+}\int_{\R^d}dy\,G(\xi^+,\eta^-\cup y)e^{-E(y,\xi^+)}e_\lambda(e^{-\phi(y-\cdot)}-1,\eta^+\setminus\xi^+),\nonumber
\end{align*}
and thus, for $G=e_\lambda(\theta^+,\theta^-)$, $\theta^\pm\in L^1$,
\begin{align*}
&(\widehat Le_\lambda(\theta^+,\theta^-))(\eta^+,\eta^-)
=-m\left(|\eta^+|+|\eta^-|\right)e_\lambda(\theta^+,\theta^-;\eta^+,\eta^-)\\
&\qquad\qquad+z\int_{\R^d}dx\,\theta^+(x)e_\lambda(\theta^+,\theta^-e^{-\phi(x-\cdot)}+e^{-\phi(x-\cdot)}-1;\eta^+,\eta^-)\\
&\qquad\qquad+z\int_{\R^d}dy\,\theta^-(y)e_\lambda(\theta^+e^{-\phi(y-\cdot)}+e^{-\phi(y-\cdot)}-1,\theta^-;\eta^+,\eta^-),
\end{align*}
where we have used the equality shown in \cite[Proposition 5.10]{KKO2002},
    \begin{multline}
        \sum_{\xi^\pm\subseteq\eta^\pm}e_\lambda(\theta^\pm,\xi^\pm)e^{-E(x,\xi^\pm)}
        e_\lambda(e^{-\phi(x-\cdot)}-1,\eta^\pm\setminus\xi^\pm)\\
        = e_\lambda(\theta^\pm e^{-\phi(x-\cdot)}+e^{-\phi(x-\cdot)}-1,\eta^\pm).\label{Equation8}
    \end{multline}
In this way,
\begin{align*}
&(\widetilde LB)(\theta^+,\theta^-)\\
=&\int_{\Gamma _0^2}d\lambda^2(\eta^+,\eta^-)\,(\widehat Le_\lambda(\theta^+,\theta^-))(\eta^+,\eta^-)k(\eta^+,\eta^-)\\
=&-m\int_{\Gamma _0^2}d\lambda^2(\eta^+,\eta^-)\,\left(|\eta^+|+|\eta^-|\right)e_\lambda(\theta^+,\theta^-;\eta^+,\eta^-)k(\eta^+,\eta^-)\\
&+z\int_{\R^d}dx\,\theta^+(x)B(\theta^+,\theta^-e^{-\phi(x-\cdot)}+e^{-\phi(x-\cdot)}-1)\\
&+z\int_{\R^d}dy\,\theta^-(y)B(\theta^+e^{-\phi(y-\cdot)}+e^{-\phi(y-\cdot)}-1,\theta^-)
\end{align*}
with
\begin{align*}
&\int_{\Gamma _0^2}d\lambda^2(\eta^+,\eta^-)\,\left(|\eta^+|+|\eta^-|\right)e_\lambda(\theta^+,\theta^-;\eta^+,\eta^-)k(\eta^+,\eta^-)\\
=&\int_{\Gamma _0}d\lambda(\eta^+)\,\sum_{x\in\eta^+}\theta^+(x)e_\lambda(\theta^+,\eta^+\setminus x)\int_{\Gamma _0}d\lambda(\eta^-)e_\lambda(\theta^-,\eta^-)k(\eta^+,\eta^-)\\
&+\int_{\Gamma _0}d\lambda(\eta^-)\,\sum_{y\in\eta^-}\theta^-(y)
e_\lambda(\theta^-,\eta^-\setminus y)\int_{\Gamma _0}d\lambda(\eta^+)e_\lambda(\theta^+,\eta^+)k(\eta^+,\eta^-),
\end{align*}
which, by \eqref{Equation3}, is equal to
\[
\int_{\R^d}\theta^+(x)\delta B(\theta^+,\theta^-; x,\emptyset)
+\int_{\R^d}\theta^-(y)\delta B(\theta^+,\theta^-;\emptyset,y).\qedhere
\]
\end{proof}

\begin{proposition}\label{Proposition2}
Let $0<\alpha<\alpha_0$ be given. If $B\in\mathcal{E}_{\alpha''}$ for some
$\alpha''\in \left(\alpha,\alpha_0\right]$, then
$\widetilde LB\in\mathcal{E}_{\alpha'}$ for all $\alpha'>0$ such that
$\alpha\leq\alpha'<\alpha''$, and we have
\[
\|\widetilde LB\|_{\alpha'}\leq
\frac{2\alpha_0}{\alpha''-\alpha'}
\Bigl(m +z\alpha_0e^{\frac{\llvert\phi\rrvert_1}{\alpha}-1}\Bigr)
\|B\|_{\alpha''}.
\]
\end{proposition}

In order to prove this result as well as other forthcoming ones the next two
lemmata show to be useful. They extend to the two-component case lemmata 3.3
and 3.4 shown in \cite{FKO2011a}.

\begin{lemma}\label{Lemma1} Given an $\alpha>0$, for all
$B\in\mathcal{E}_\alpha$ let
\begin{align*}
&(L_0^+B)(\theta^+,\theta^-):=\int_{\R^d}dx\,\theta^+(x)\vd B(\theta^+,\theta^-; x,\emptyset),\\
&(L_0^-B)(\theta^+,\theta^-):=\int_{\R^d}dy\,\theta^-(y)\vd B(\theta^+,\theta^-; \emptyset,y),\quad
\theta^\pm\in L^1.
\end{align*}
Then, for all $\alpha'<\alpha$, we have $L_0^\pm B\in\mathcal{E}_{\alpha'}$ and,
moreover, the following estimate of norms holds:
\[
\|L_0^\pm B\|_{\alpha'}\leq \frac{\alpha'}{\alpha-\alpha'}\|B\|_{\alpha}.
\]
\end{lemma}

\begin{proof} First we observe that, by Subsection \ref{Subsection23},
$L_0^\pm B$ are entire functionals on $L^1\times L^1$ and, moreover,
for all $r>0$ and all $\theta^\pm\in L^1$,
\begin{align*}
\left|(L_0^+B)(\theta^+,\theta^-)\right|\leq&\,\llvert\theta^+\rrvert_1 \left\|\vd
B(\theta^+,\theta^-;\cdot,\emptyset)\right\|_{L^\infty(\R^d)}\\
\leq&\,\frac{\llvert\theta^+\rrvert_1 }{r}\sup_{\llvert\theta_0^+\rrvert_1 \leq
r}\left|B(\theta^++\theta_0^+,\theta^-)\right|,
\end{align*}
where, for all $\theta_0^+\in L^1$ such that $\llvert\theta_0^+\rrvert_1 \leq r$,
\[
\left|B(\theta^++\theta_0^+,\theta^-)\right|\leq
\|B\|_\alpha e^{\frac{1}{\alpha}(\llvert\theta^+\rrvert_1 +\llvert\theta^-\rrvert_1 +r)}.
\]
Hence,
\begin{align*}
\|L_0^+B\|_{\alpha'}=&\sup_{\theta^\pm\in L^1}\left(e^{-\frac{1}{\alpha'}(\llvert\theta^+\rrvert_1 +\llvert\theta^-\rrvert_1 )}|(L_0^+ B)(\theta^+,\theta^-)|\right)\\
\leq&\,\frac{e^{\frac{r}{\alpha}}}{r}\sup_{\theta^\pm\in L^1}\left(e^{-\left(\frac{1}{\alpha'}-\frac{1}{\alpha}\right)(\llvert\theta^+\rrvert_1 +\llvert\theta^-\rrvert_1 )}\llvert\theta^+\rrvert_1 \right)\|B\|_\alpha,
\end{align*}
where the latter supremum is finite if and only if $\frac{1}{\alpha'}-\frac{1}{\alpha}>0$. In such
a situation, the use of the inequalities $xe^{-n(x+y)}\leq xe^{-nx}\leq\frac{1}{en}$, $x,y\geq0$, $n>0$ leads for
each $r>0$ to
\begin{equation*}
\|L_0^+B\|_{\alpha'}\leq\frac{e^{\frac{r}{\alpha}}}{r}\frac{\alpha\alpha'}{e(\alpha-\alpha')}\|B\|_\alpha.
\end{equation*}
Analogously to \cite[Lemma 3.3]{FKO2011a}, the required estimate of norms follows
by minimizing the expression
$\frac{e^{\frac{r}{\alpha}}}{r}\frac{\alpha\alpha'}{e(\alpha-\alpha')}$
in the parameter $r$. Similar arguments applied to $L_0^-B$ completes the
proof.
\end{proof}

\begin{lemma}\label{Lemma2} Let $\varphi,\psi:\R^d\times\R^d\to\R$ be such that,
for a.a.~$\omega\in\R^d$,
$\varphi(\omega,\cdot)\in L^\infty:=L^\infty(\R^d)$,
$\psi(\omega,\cdot)\in L^1$ and $\|\varphi(\omega,\cdot)\|_{L^\infty}\leq c_0$,
$\llvert\psi(\omega,\cdot)\rrvert_1 \leq c_1$ for some constants $c_0,c_1>0$ independent
of $\omega$. For each $\alpha>0$ and all $B\in\mathcal{E}_\alpha$, consider
\begin{align*}
&(L_1^+B)(\theta^+,\theta^-):=\int_{\R^d}dx\,\theta^+(x)B(\theta^+,\varphi(x, \cdot)\theta^-
+\psi(x,\cdot)),\\
&(L_1^-B)(\theta^+,\theta^-):=\int_{\R^d}dy\,\theta^-(y)B(\varphi(y, \cdot)\theta^+
+\psi(y,\cdot), \theta^-),\quad\theta^\pm\in L^1.
\end{align*}
Then, for all $\alpha'>0$ such that $c_0\alpha'<\alpha$, we have
$L_1^\pm B\in\mathcal{E}_{\alpha'}$
and
\begin{equation}
\|L_1^\pm B\|_{\alpha'}\leq \frac{\alpha\alpha'}{\alpha-\alpha'}e^{\frac{c_1}{\alpha}-1}\|B\|_{\alpha}.
\label{Equation1}
\end{equation}
\end{lemma}

\begin{proof} As before, the entireness property of $L_1^+B$ and $L_1^-B$ on
$L^1\times L^1$ follows from Subsection \ref{Subsection23}. In this way, given
a $B\in\mathcal{E}_\alpha$, for all $\theta^\pm\in L^1$ one has
\[
|B(\theta^+,\varphi(x,\cdot)\theta^-+\psi(x,\cdot))|\leq \|B\|_{\alpha}\,
e^{\frac{1}{\alpha}\left(\llvert\theta^+\rrvert_1 +c_0\llvert\theta^-\rrvert_1 +c_1\right)},
\]
which implies
\begin{align*}
\|L_1^+B\|_{\alpha'}
&\leq\sup_{\theta^\pm\in L^1}\biggl(e^{-\frac{1}{\alpha'}(\llvert\theta^+\rrvert_1 +\llvert\theta^-\rrvert_1 )}\int_{\R^d}dx\,\left|\theta^+(x)B(\theta^+,\varphi(x,\cdot)\theta^-+\psi(x,\cdot))\right|\biggr)\\
&\leq e^{\frac{c_1}{\alpha}}\|B\|_{\alpha}\sup_{\theta^\pm\in L^1}
\left(e^{-\left(\frac{1}{\alpha'}-\frac{1}{\alpha}\right)\llvert\theta^+\rrvert_1 -\left(\frac{1}{\alpha'}-\frac{c_0}{\alpha}\right)\llvert\theta^-\rrvert_1 }\llvert\theta^+\rrvert_1 \right).
\end{align*}
Concerning the latter supremum, observe that it is finite provided
$\frac{1}{\alpha'}-\frac{1}{\alpha}>0$ and
$\frac{1}{\alpha'}-\frac{c_0}{\alpha}>0$. In this case, the use of the
inequality $xe^{-m_1x-m_2y}\leq xe^{-m_1x}$, $x,y\geq 0$, $m_1,m_2>0$ allows us
to proceed by arguments similar to those used in the previous lemma. A
similar proof yields the estimate of norms \eqref{Equation1} for $L_1^-B$.
\end{proof}

\begin{proof}[Proof of Proposition \ref{Proposition2}]
In Lemma \ref{Lemma2} replace $\varphi$ by $e^{-\phi}$ and $\psi$ by
$e^{-\phi}-1$. Due to the positiveness and integrability
properties of $\phi$ one has $e^{-\phi}\leq 1$ and
$|e^{-\phi}-1|=1-e^{-\phi}\leq \phi\in L^1$, ensuring the conditions
to apply Lemma \ref{Lemma2}. This combined with an application of Lemma
\ref{Lemma1} leads to the required estimate of norms.
\end{proof}

As a consequence of Proposition \ref{Proposition2}, one may state the next
existence and uniqueness result. Its proof follows as a particular
application of an Ovsjannikov-type result in a scale of Banach spaces
$\{\mathcal{E}_\alpha:0<\alpha\leq\alpha_0\}$, $\alpha_0>0$, defined in
Subsection \ref{Subsection23}. For convenience of the reader, this statement
is recalled in Appendix below (Theorem \ref{Th1}).

\begin{theorem}\label{Theorem1}
Given an $\alpha_0>0$, let
$B_0\in\mathcal{E}_{\alpha_0}$. For each $\alpha\in(0,\alpha_0)$ there is a
$T>0$, that is,
\[
T=\biggl(2e\alpha_0\Bigl(m+z\alpha_0e^{\frac{\llvert\phi\rrvert_1 }{\alpha}-1}\Bigr)\biggr)^{-1}(\alpha_0-\alpha),
\]
such that there is a unique solution $B_t$, $t\in[0,T)$, to the initial value
problem $\dfrac{\partial B_t}{\partial t}=\widetilde LB_t$, ${B_t}_{|t=0}= B_0$ in the
space $\mathcal{E}_\alpha$.
\end{theorem}

\section{Lebowitz--Penrose-type scaling}\label{Section4}

In the lattice case, one of the basic questions in the theory of Ising models with long
range interactions is the investigation of the behavior of the system as the
range of the interaction increases to infinity, see e.g.~\cite{KL1999,Pre2009}.
In this section we extend this investigation to a continuous particle system,
namely, to the Widom--Rowlinson model.

By analogy with the lattice case, the starting point is the scale
transformation $\phi\mapsto\eps^d\phi(\eps\cdot)$,
$\eps >0$, of the operator $L$, that is\footnote{Here and below, for simplicity of notation, we have just written $\eps\gamma^\pm$, instead of $\{\eps x: x\in\gamma^\pm\}$.
In the sequel, we also will use the notation ${\eps^{-1}}{\eta^\pm}$ for the set $\{{\eps^{-1}}{x}: x\in\eta^\pm\}$.},
\begin{align*}
(L_\eps F)(\gamma^+,\gamma^-):=&\,m\sum_{x\in \gamma^+}\left(F(\gamma^+\setminus x,\gamma^-) -
F(\gamma^+,\gamma^-)\right)\\
&+m\sum_{y\in \gamma^-}\left(F(\gamma^+,\gamma^-\setminus y) - F(\gamma^+,\gamma^-)\right)\\
&+ z\int_{\R^d} dx\,e^{-\eps^d E(\eps x,\eps\gamma^-)} \left(F(\gamma^+\cup x,\gamma^-) - F(\gamma^+,\gamma^-)\right)\\
&+z\int_{\R^d} dy\,e^{-\eps^dE(\eps y,\eps\gamma^+)} \left(F(\gamma^+,\gamma^-\cup y) - F(\gamma^+,\gamma^-)\right).
\end{align*}
As explained before, in terms of correlation functions this yields an initial
value problem
\begin{equation}
\frac \partial {\partial t}k_t^{(\eps)}=\widehat L_\eps^*k^{(\eps)}_t,\quad
{k_t^{(\eps)}}_{|t=0}=k_0^{(\eps)},\label{Equation4}
\end{equation}
for a proper scaled $k_0^{(\eps)}$ initial correlation function, which
corresponds to a compressed initial particle system. More precisely, we
consider the following mapping
\[
(S_\eps k)(\eta^+,\eta^-):=k(\eps\eta^+,\eps\eta^-),\quad \eps>0,
\]
and we choose a singular initial correlation function $k_0^{(\eps)}$ such that
its renormalization $k^{(\eps)}_{0,\mathrm{ren}}:=S_{\eps^{-1}}k_0^{(\eps)}$
converges pointwisely as $\eps$ tends to zero to a function which is
independent of $\eps$. This leads then to a renormalized version of the
initial value problem \eqref{Equation4},
\begin{equation}\label{lalala}
\frac \partial {\partial t}k_{t,\mathrm{ren}}^{(\eps)}={\widehat L_{\eps,\mathrm{ren}}}^*k^{(\eps)}_{t,\mathrm{ren}},\quad
{k_{t,\mathrm{ren}}^{(\eps)}}_{|t=0}=k_{0,\mathrm{ren}}^{(\eps)},
\end{equation}
with ${\widehat L_{\eps,\mathrm{ren}}}^*=S_{\eps^{-1}}{\widehat L_\eps}^*S_{\eps}$,
cf.~\cite{FKK2010a}. Clearly,
\[
k^{(\eps)}_{t,\mathrm{ren}}(\eta^+,\eta^-)=(S_{\eps^{-1}} k_t^{(\eps)})(\eta^+,\eta^-)
=k_t^{(\eps)}\left({\eps^{-1}}{\eta^+},{\eps^{-1}}{\eta^-}\right),
\]
provided solutions to \eqref{Equation4} and to \eqref{lalala} exist.

In terms of GF, this scheme yields
\[
B_{t,\mathrm{ren}}^{(\eps)}(\theta^+,\theta^-):=\int_{\Gamma_0^2}d\lambda^2(\eta^+,\eta^-)\,e_\lambda(\theta^+,\theta^-;\eta^+,\eta^-)k_{t,\mathrm{ren}}^{(\eps)}(\eta^+,\eta^-)
\]
leading, as in \eqref{obtevol}, to the initial value problem
\begin{equation}
\frac{\partial}{\partial t}B_{t,\mathrm{ren}}^{(\eps)}
=\widetilde L_{\eps, \mathrm{ren}}B_{t,\mathrm{ren}}^{(\eps)},\quad
{B_{t,\mathrm{ren}}^{(\eps)}}_{|t=0}= B_{0,\mathrm{ren}}^{(\eps)}\label{V12}
\end{equation}
with
\[
(\widetilde L_{\eps,\mathrm{ren}}B)(\theta^+,\theta^-)=
\int_{\Gamma _0^2}d\lambda^2(\eta^+,\eta^-)\,(\widehat L_{\eps,\mathrm{ren}}e_\lambda(\theta^+,\theta^-))(\eta^+,\eta^-)k(\eta^+,\eta^-).
\]
Here, by a dual relation like the one in \eqref{duality},
$\widehat L_{\eps, \mathrm{ren}}=S_{\eps}^*\widehat L_\eps S_{\eps^{-1}}^*$ with
\begin{align}
(S_\eps^*G)(\eta^+,\eta^-)=&\,\eps^{-d(|\eta^+|+|\eta^-|)}
G\left({\eps^{-1}}{\eta^+},{\eps^{-1}}{\eta^-}\right),\label{Equation5}\\
(S_{\eps^{-1}}^*G)(\eta^+,\eta^-)=&\,\eps^{d(|\eta^+|+|\eta^-|)}
G(\eps\eta^+,\eps\eta^-).\label{Equation7}
\end{align}

In the sequel we fix the notation
\begin{equation}\label{psieps}
\psi_\eps(x)= \frac{e^{-\eps^d\phi (x)}-1}{\eps^d}, \qquad x\in\R^d,\  \eps>0.
\end{equation}

\begin{proposition}\label{Proposition3}
For all $\eps >0$ and all $\theta^\pm\in L^1$, we have
\begin{align*}
&(\widetilde L_{\eps,\mathrm{ren}}B)(\theta^+,\theta^-)\\
=&-\int_{\R^d}dx\,\theta^+(x)\left(m\delta B(\theta^+,\theta^-; x,\emptyset)
-zB\left(\theta^+,\theta^- e^{-\eps^d\phi (x -\cdot )}+\psi_\eps(x-\cdot)\right)\right)\\
&-\int_{\R^d}dy\,\theta^-(y)\left(m\delta B(\theta^+,\theta^-;\emptyset, y)
-zB\left(\theta^+ e^{-\eps^d\phi (y -\cdot )}+\psi_\eps(y-\cdot),\theta^-\right)\right).
\end{align*}
\end{proposition}

\begin{proof} Similarly to the proof of Proposition \ref{Proposition1},
one obtains the following explicit form for
$\widehat L_\eps:=\KK^{-1}L_\eps\KK$,
\begin{align}
&(\widehat L_\eps G)(\eta^+,\eta^-)\label{Equation6}\\
=&-m\left(|\eta^+|+|\eta^-|\right)G(\eta^+,\eta^-)\nonumber\\
&+z\sum_{\xi^-\subseteq\eta^-}\int_{\R^d}dx\,G(\eta^+\cup x,\xi^-)e^{-\eps^dE(\eps x,\eps\xi^-)}e_\lambda(e^{-\eps^d\phi(\eps(x-\cdot))}-1,\eta^-\setminus\xi^-)\nonumber\\
&+z\sum_{\xi^+\subseteq\eta^+}\int_{\R^d}dy\,G(\xi^+,\eta^-\cup y)
e^{-\eps^d E(\eps y,\eps\xi^+)}e_\lambda(e^{-\eps^d\phi(\eps(y-\cdot))}-1,\eta^+\setminus\xi^+).\nonumber
\end{align}
Therefore, for any $\theta^\pm\in L^1$, it follows from
\eqref{Equation5} and \eqref{Equation6}
\[
(\widehat L_{\eps,\mathrm{ren}}e_\lambda(\theta^+,\theta^-))(\eta^+,\eta^-)\\
=\eps^{-d(|\eta^+|+|\eta^-|)}(\widehat L_\eps S_{\eps^{-1}}^*
e_\lambda(\theta^+,\theta^-))\left({\eps^{-1}}{\eta^+},{\eps^{-1}}{\eta^-}\right)
\]
with
\begin{align*}
&(\widehat L_\eps S_{\eps^{-1}}^*
e_\lambda(\theta^+,\theta^-))\left({\eps^{-1}}{\eta^+},{\eps^{-1}}{\eta^-}\right)\\
=&-m\left(\left|{\eps^{-1}}{\eta^+}\right|+\left|{\eps^{-1}}{\eta^-}\right|\right)(S_{\eps^{-1}}^*
e_\lambda(\theta^+,\theta^-))\left({\eps^{-1}}{\eta^+},{\eps^{-1}}{\eta^-}\right)\\
&+z\!\!\!\sum_{\xi^-\subseteq{\eps^{-1}}{\eta^-}}\int_{\R^d}dx\,
(S_{\eps^{-1}}^* e_\lambda(\theta^+,\theta^-))\left(({\eps^{-1}}{\eta^+})\cup x,\xi^-\right)\\&\qquad\qquad\qquad\times e^{-\eps^dE(\eps x,\eps\xi^-)}e_\lambda\left(e^{-\eps^d\phi(\eps(x-\cdot))}-1,({\eps^{-1}}{\eta^-})
\setminus\xi^-\right)\\
&+z\!\!\!\sum_{\xi^+\subseteq{\eps^{-1}}{\eta^+}}\int_{\R^d}dy\,(S_{\eps^{-1}}^* e_\lambda(\theta^+,\theta^-))\left(\xi^+,({\eps^{-1}}{\eta^-})\cup y\right)\\&\qquad\qquad\qquad\times
e^{-\eps^d E(\eps y,\eps\xi^+)}e_\lambda\left(e^{-\eps^d\phi(\eps(y-\cdot))}-1,({\eps^{-1}}{\eta^+})
\setminus\xi^+\right),
\end{align*}
where $|{\eps^{-1}}{\eta^\pm}|=|\eta^\pm|$. Moreover, for a generic
function $G$ one has
\[
\sum_{\xi\subseteq{\eps^{-1}}{\eta}}G(\xi)=\sum_{\xi\subseteq\eta}G\left({\eps^{-1}}{\xi}\right),
\]
allowing us to rewrite the latter equality as
\begin{align*}
&(\widehat L_\eps S_{\eps^{-1}}^*
e_\lambda(\theta^+,\theta^-))\left({\eps^{-1}}{\eta^+},{\eps^{-1}}{\eta^-}\right)\\
=&-m(|\eta^+|+|\eta^-|)(S_{\eps^{-1}}^*
e_\lambda(\theta^+,\theta^-))\left({\eps^{-1}}{\eta^+},{\eps^{-1}}{\eta^-}\right)\\
&+z\!\!\!\sum_{\xi^-\subseteq\eta^-}\int_{\R^d}dx\,
(S_{\eps^{-1}}^* e_\lambda(\theta^+,\theta^-))\left(({\eps^{-1}}{\eta^+})\cup x,{\eps^{-1}}{\xi^-}\right)\\&\qquad\qquad\qquad\times e^{-\eps^d
E(\eps x,\xi^-)}
e_\lambda\left(e^{-\eps^d\phi(\eps x-\cdot)}-1,\eta^-\setminus\xi^-\right)\\
&+z\!\!\!\sum_{\xi^+\subseteq\eta^+}\int_{\R^d}dy\,(S_{\eps^{-1}}^* e_\lambda(\theta^+,\theta^-))\left({\eps^{-1}}{\xi^+},({\eps^{-1}}{\eta^-})\cup y\right)\\&\qquad\qquad\qquad\times
e^{-\eps^d E(\eps y,\xi^+)}e_\lambda\left(e^{-\eps^d\phi(\eps y-\cdot)}-1,\eta^+\setminus\xi^+\right).
\end{align*}
As a result, since by \eqref{Equation7}
\begin{align*}
(S_{\eps^{-1}}^*e_\lambda(\theta^+,\theta^-))\left({\eps^{-1}}{\eta^+},{\eps^{-1}}{\eta^-}\right)=&\,\eps^{d(|\eta^+|+|\eta^-|)}
e_\lambda(\theta^+,\theta^-;\eta^+,\eta^-),\\
(S_{\eps^{-1}}^*e_\lambda(\theta^+,\theta^-))\left(({\eps^{-1}}{\eta^+})\cup
x,{\eps^{-1}}{\xi^-}\right)=&\,\eps^{d(|\eta^+|+|\xi^-|+1)}
e_\lambda(\theta^+,\theta^-;\eta^+\cup\eps x,\xi^-),\\
(S_{\eps^{-1}}^*e_\lambda(\theta^+,\theta^-))\left({\eps^{-1}}{\xi^+},({\eps^{-1}}{\eta^-})\cup
y\right)=&\,\eps^{d(|\xi^+|+|\eta^-|+1)}
e_\lambda(\theta^+,\theta^-;\xi^+,\eta^-\cup\eps y),
\end{align*}
a change of variables, $\eps x\mapsto \omega_1$,
$\eps y\mapsto \omega_2$, followed by an application of equality
\eqref{Equation8} lead at the end to
\begin{align*}
&(\widehat L_{\eps,\mathrm{ren}}e_\lambda(\theta^+,\theta^-))(\eta^+,\eta^-)\\
=&-m(|\eta^+|+|\eta^-|)e_\lambda(\theta^+,\theta^-;\eta^+,\eta^-)\\
&+z\int_{\R^d}dx\,\theta^+(x)e_\lambda\left(\theta^+,\theta^- e^{-\eps^d\phi (x -\cdot )}+\psi_\eps(x-\cdot);\eta^+,\eta^-\right)\\
&+z\int_{\R^d}dy\,\theta^-(y)e_\lambda\left(\theta^+ e^{-\eps^d\phi (y -\cdot )}+\psi_\eps(y-\cdot),\theta^-;\eta^+,\eta^-\right),
\end{align*}
where $\psi_\eps$ is the function defined in \eqref{psieps}.

Similar arguments used to prove Proposition \ref{Proposition1} then yield
the required expression for the operator $\widetilde L_{\eps,\mathrm{ren}}$.
\end{proof}

\begin{remark}
The proof of Proposition \ref{Proposition3} yields an explicit
form for the operator $\widehat L_{\eps,\mathrm{ren}}$. One can show that the
mesoscopic scaling in the sense of Vlasov, cf.~\cite{FKK2010a,FKK2010b}, gives the same expression for the corresponding operator $\widehat L_{\eps,\mathrm{ren}}$.
\end{remark}

\begin{proposition}\label{Proposition4}
(i) If $B\in\mathcal{E}_{\alpha}$ for
some $\alpha >0$, then, for all $\theta^\pm\in L^1$,
$(\widetilde L_{\eps,\mathrm{ren}}B)(\theta^+,\theta^-)$ converges as $\eps$
tends to zero to
\begin{align*}
&(\widetilde L_{LP} B)(\theta^+,\theta^-)\\
:=&-\int_{\R^d}dx\,\theta^+(x)\left(m\delta B(\theta^+,\theta^-;x,\emptyset) -zB(\theta^+,\theta^--\phi(x-\cdot))\right)\\
&-\int_{\R^d}dy\,\theta^-(y)\left(m\delta B(\theta^+,\theta^-;\emptyset,y) -zB(\theta^+-\phi(y-\cdot),\theta^-)\right).
\end{align*}
(ii) Let $\alpha_0>\alpha>0$ be given. If $B\in\mathcal{E}_{\alpha''}$ for some
$\alpha''\in (\alpha, \alpha_0]$, then $\widetilde L_{\eps, \mathrm{ren}}B,\widetilde{L}_{LP} B \in\mathcal{E}_{\alpha'}$ for all $\alpha'>0$ such that
$\alpha\leq\alpha'<\alpha''$. Moreover,
\[
\|\widetilde{L}_{\#}B\|_{\alpha'}\leq
\frac{2\alpha_0}{\alpha''-\alpha'}
\Bigl(m+z\alpha_0e^{\frac{\llvert\phi\rrvert_1}{\alpha}-1}\Bigr)
\|B\|_{\alpha''},
\]
where $\widetilde{L}_{\#}$ denotes either $\widetilde{L}_{\eps, \mathrm{ren}}$ or
$\widetilde{L}_{LP}$.
\end{proposition}

\begin{proof} (i) Given a $\theta\in L^1$, observe that for
a.a.~$\omega\in\R^d$ one clearly has
\[
\lim_{\eps\searrow 0}\left(\theta e^{-\eps^d\phi (\omega -\cdot )}+\psi_\eps(\omega-\cdot)\right)=\theta -\phi(\omega-\cdot)\ \mbox{in}\ L^1.
\]
Hence, due to the continuity of the functionals $B(\theta^+,\cdot)$ and
$B(\cdot,\theta^-)$ in $L^1$ (both are even entire on $L^1$) the following
limits hold a.e.
\begin{equation}
\begin{aligned}
&\lim_{\eps\searrow 0}B\Bigl(\theta^+,\theta^- e^{-\eps^d\phi (x -\cdot )}+\psi_\eps(x-\cdot)\Bigr)=B\bigl(\theta^+,\theta^- -\phi(x-\cdot)\bigr),\\
&\lim_{\eps\searrow 0}B\Bigl(\theta^+ e^{-\eps^d\phi (y -\cdot )}+\psi_\eps(y-\cdot),\theta^-\Bigr)=B\bigl(\theta^+ -\phi(y-\cdot),\theta^-\bigr),
\end{aligned}\label{dop-eq}
\end{equation}
showing the pointwise convergence of the integrand functions which
appear in the definition of
$(\widetilde L_{\eps, \mathrm{ren}}B)(\theta^+,\theta^-)$ and
$(\widetilde L_{LP} B)(\theta^+,\theta^-)$. Moreover, since the absolute value
of both expressions appearing in the left-hand side of \eqref{dop-eq} are
bounded, for all $\eps>0$, by
\[
\|B\|_\alpha\exp\biggl(\frac{1}{\alpha}\left(\llvert\theta^+\rrvert_1 +\llvert\theta^-\rrvert_1 +\llvert\phi\rrvert_1\right)\biggr),
\]
an application of the Lebesgue dominated convergence theorem leads then to the required limit.

(ii) Both estimates of norms follow as a particular application of Lemmata
\ref{Lemma1} and \ref{Lemma2}. For the case of
$\widetilde{L}_{\eps, \mathrm{ren}}B$, by replacing in Lemma \ref{Lemma2}
$\varphi$ by $e^{-\eps^d\phi}$ and $\psi$ by $\psi_\eps$, defined in
\eqref{psieps}, for the case of $\widetilde{L}_{LP}B$, by replacing in Lemma
\ref{Lemma2} $\varphi$ by the function identically equal to 1 and $\psi$ by
$-\phi$. Due to the positiveness and integrability assumptions on $\phi$ the
proof follows similarly to the proof of Proposition \ref{Proposition2}.
\end{proof}

Proposition \ref{Proposition4} (ii) provides similar estimate of norms for
$\widetilde{L}_{\eps, \mathrm{ren}}$, $\eps>0$, and the limiting mapping
$\widetilde{L}_{LP}$, namely, $\|\widetilde{L}_{\eps, \mathrm{ren}}B\|_{\alpha'},\|\widetilde{L}_{LP}B\|_{\alpha'}\leq \frac{M}{\alpha''-\alpha'}\|B\|_{\alpha''}$, $0<\alpha\leq\alpha'<\alpha''\leq\alpha_0$, with
\[
M:=2\alpha_0\left(m+z\alpha_0e^{\frac{\llvert\phi\rrvert_1}{\alpha}-1}\right).
\]
Therefore, given any
$B_{0,LP},B_{0,\mathrm{ren}}^{(\eps)}\in\mathcal{E}_{\alpha_0}$, $\eps>0$, it
follows from Theorem \ref{Th1} that for each
$\alpha\in\left(0,\alpha_0\right)$ and $\delta=\frac{1}{eM}$ there is a unique solution
$B_{t,\mathrm{ren}}^{(\eps)}:\left[0,\delta(\alpha_0-\alpha)\right)\to\mathcal{E}_\alpha$,
$\eps>0$, to each initial value problem \eqref{V12} and a unique solution
$B_{t,LP}:\left[0,\delta(\alpha_0-\alpha)\right)\to\mathcal{E}_\alpha$
to the initial value problem
\begin{equation}
\frac{\partial}{\partial t}B_{t,LP}=\widetilde L_{LP}B_{t,LP},\quad
{B_{t,LP}}_{|t=0}=B_{0,LP}.\label{V19}
\end{equation}
That is, independent of the initial value problem under consideration, the
solutions obtained are defined on the same time-interval and with values in
the same Banach space. Therefore, it is natural to analyze under which
conditions the solutions to \eqref{V12} converge to the solution to
\eqref{V19}. These conditions are stated in Theorem \ref{Theorem2} below and
they follow from the next result and a particular application of
\cite[Theorem 4.3]{FKO2011a}, recalled in Appendix below (Theorem
\ref{Thconv}).

\begin{proposition}\label{Proposition5}
Assume that $0\leq\phi\in L^1\cap L^\infty$ and let $\alpha_0>\alpha>0$ be
given. Then, for all $B\in\mathcal{E}_{\alpha''}$,
$\alpha''\in (\alpha, \alpha_0]$, the following estimate holds
 \begin{equation*}
\|\widetilde L_{\eps, \mathrm{ren}}B-\widetilde L_{LP} B\|_{\alpha'}\leq
2\eps^d z\|\phi\|_{L^\infty}e^{\frac{\llvert\phi\rrvert_1}{\alpha}}\left(
  \frac{\alpha_0}{\alpha''-\alpha'}\llvert\phi\rrvert_1 +\frac{\alpha_0^3}{(\alpha''-\alpha')^2e}\right)\|B\|_{\alpha''},
\end{equation*}
for all $\alpha'$ such that $\alpha\leq\alpha'<\alpha''$ and all $\eps>0$.
\end{proposition}

\begin{proof} Since
\begin{align}
&\left|(\widetilde L_{\eps, \mathrm{ren}}B)(\theta^+,\theta^-)-(\widetilde L_{LP} B)(\theta^+,\theta^-)\right|\nonumber\\
\leq&\,z\int_{\R^d}\!dx\,\left|\theta^+(x)\right|\nonumber\\
&\quad \times\left|B\!\!\left(\theta^+\!,\theta^- e^{-\eps^d\phi (x -\cdot )}+\psi_\eps(x -\cdot
)\!\right)\!-\!B\left(\theta^+,\theta^-\! -\!\phi (x -\cdot)\right)\right|\label{Ola}\\
&+z\!\!\int_{\R^d}\!dy\,\left|\theta^-(y)\right|\nonumber\\
&\quad \times\left|B\!\!\left(\theta^+ e^{-\eps^d\phi (y -\cdot )}+\psi_\eps (y -\cdot
),\theta^-\!\!\right)\!-\!B\left(\theta^+\! -\!\phi (y -\cdot),\theta^-\right)\right|,\label{Equation9}
\end{align}
first we will estimate \eqref{Ola}. For this purpose, given any
$\theta^+,\theta_1^-,\theta_2^-\in L^1$, let us consider the function
$C_{\theta^+,\theta_1^-,\theta_2^-}(t)=B\left(\theta^+,t\theta_1^-+(1-t)\theta_2^-\right)$,
$t\in\left[0,1\right]$. Hence
\begin{align*}
\frac{\partial}{\partial
t} C_{\theta^+,\theta_1^-,\theta_2^-}(t)=&\,\frac{\partial}{\partial
s} C_{\theta^+,\theta_1^-,\theta_2^-}(t+s)\Bigr|_{s=0}
\\=&\,\frac{\partial}{\partial
s} B\left(\theta^+,\theta_2^-+t(\theta_1^--\theta_2^-)+s(\theta_1^--\theta_2^-)\right)\Bigr|_{s=0}
\\=&\,\int_{\mathbb{R}^d}dy\,(\theta_1^-(y)-\theta_2^-(y))\,\vd B(\theta^+,\theta_2^-+t(\theta_1^--\theta_2^-);\emptyset,y),
\end{align*}
which leads to
\begin{align*}
&\left|B(\theta^+,\theta_1^-)-B(\theta^+,\theta_2^-)\right|
=\left|C_{\theta^+,\theta_1^-,\theta_2^-}(1)-C_{\theta^+,\theta_1^-,\theta_2^-}(0)\right|\\
\leq&\max_{t\in[0,1]}\int_{\mathbb{R}^d}dy\,\left|\theta_1^-(y)-\theta_2^-
(y)\right| \left|\vd B(\theta^+,\theta_2^-+t(\theta_1^--\theta_2^-);\emptyset, y)\right|\\
\leq&\,\llvert\theta_1^--\theta_2^-\rrvert_1 \max_{t\in[0,1]}\|\vd B(\theta^+,\theta_2^-+t(\theta_1^--\theta_2^-);\emptyset,\cdot)\|_{L^\infty},
\end{align*}
where, by similar arguments used to prove Lemma \ref{Lemma1},
\begin{align*}
&\left\|\vd B(\theta^+,\theta_2^-+t(\theta_1^--\theta_2^-);\emptyset,\cdot)\right\|_{L^\infty}\\
\leq&\frac{e}{\alpha''}\exp\left(\frac{\llvert\theta^+\rrvert_1 +\llvert\theta_2^-+t(\theta_1^--\theta_2^-)\rrvert_1 }{\alpha''}\right)\|B\|_{\alpha''}.
\end{align*}
As a result
\begin{align*}
&\left|B(\theta^+,\theta_1^-)-B(\theta^+,\theta_2^-)\right|\\
\leq&\frac{e^{\frac{\llvert\theta^+\rrvert_1 }{\alpha''}+1}}{\alpha''}\llvert\theta_1^--\theta_2^-\rrvert_1 \|B\|_{\alpha''}\max_{t\in[0,1]}\exp\left(\frac{t\llvert\theta_1^-\rrvert_1 +(1-t)\llvert\theta_2^-\rrvert_1 }{\alpha''}\right),
\end{align*}
for all $\theta_1^-,\theta_2^-\in L^1$. In particular, this shows that
\begin{align*}
&\left|  B\left(\theta^+,\theta^- e^{-\eps^d\phi (x -\cdot
)}+\psi_\eps (x -\cdot
)\right)-B\left(\theta^+,\theta^- -\phi (x -\cdot)\right)\right|\\\leq&
\,\eps^d\frac{e^{\frac{\llvert\theta^+\rrvert_1 }{\alpha''}+1}}{\alpha''}\|\phi\|_{L^\infty}\left(\llvert\theta^-\rrvert_1 +
\llvert\phi\rrvert_1\right)\|B\|_{\alpha''}\\&\qquad\times
\max_{t\in[0,1]}\exp\left(\frac{1}{\alpha''}\left(t\left(\llvert\theta^-\rrvert_1 +
\llvert\phi\rrvert_1\right)+(1-t)\left(\llvert\theta^-\rrvert_1 +\llvert\phi\rrvert_1 \right)\right)\right)\\=&\,
\frac{\eps^d}{\alpha''}\|\phi\|_{L^\infty}
\left(\llvert\theta^-\rrvert_1 +\llvert\phi\rrvert_1\right)
\exp\left(\frac{1}{\alpha''}\left(\llvert\theta^+\rrvert_1 +\llvert\theta^-\rrvert_1 +\llvert\phi\rrvert_1\right)+1\right)\|B\|_{\alpha''},
\end{align*}
where we have used the inequalities
\begin{align*}
\llvert\theta^- e^{-\eps^d\phi (x -\cdot )}-\theta^-\rrvert_1 &\leq\eps^d\|\phi\|_{L^\infty}\,\llvert\theta^-\rrvert_1 ,\\
\llvert\psi_\eps (x -\cdot
)+\phi (x -\cdot
)\rrvert_1 &\leq\eps^d\|\phi\|_{L^\infty}\,\llvert\phi\rrvert_1,\\
\llvert\theta^- e^{-\eps^d\phi (x -\cdot )}+\psi_\eps (x -\cdot
)\rrvert_1 &\leq\llvert\theta^-\rrvert_1 +\llvert\phi\rrvert_1.
\end{align*}

Of course, a similar approach may also be used to estimate \eqref{Equation9}. In this case,
given any $\theta^+_1,\theta_2^+,\theta^-\in L^1$ and the function defined by
$C_{\theta^+_1,\theta_2^+,\theta^-}(t)=B\left(t\theta_1^++(1-t)\theta_2^+,\theta^-\right)$,
$t\in\left[0,1\right]$, similar arguments lead to
\begin{align*}
&\Bigl|  B\bigl(\theta^+ e^{-\eps^d\phi (y -\cdot
)}+\psi_\eps (y -\cdot
),\theta^-\bigr)-B\left(\theta^+ -\phi (y -\cdot),\theta^-\right)\Bigr|\\\leq&\,
\frac{\eps^d}{\alpha''}\|\phi\|_{L^\infty}
\left(\llvert\theta^+\rrvert_1 +\llvert\phi\rrvert_1\right)
\exp\left(\frac{1}{\alpha''}\left(\llvert\theta^+\rrvert_1 +\llvert\theta^-\rrvert_1 +\llvert\phi\rrvert_1\right)+1\right)\|B\|_{\alpha''}.
\end{align*}

As a result, from the estimates derived for \eqref{Ola} and for
\eqref{Equation9} one obtains
\begin{align*}
&\|\widetilde L_{\eps, \mathrm{ren}}B-\widetilde L_{LP} B\|_{\alpha'}\\
\leq&\,2\eps^dz\frac{e^{\frac{\llvert\phi\rrvert_1 }{\alpha''}+1}}{\alpha''}\|\phi\|_{L^\infty}\|B\|_{\alpha''}\\
&\qquad\times\sup_{\theta^\pm\in L^1}\left(\llvert\theta^+\rrvert_1 \llvert\theta^-\rrvert_1
\exp\left(-\Bigl(\frac{1}{\alpha'}-\frac{1}{\alpha''}\Bigr)(\llvert\theta^+\rrvert_1 +\llvert\theta^-\rrvert_1 )\right)\right)\\
&+\eps^dz\frac{e^{\frac{\llvert\phi\rrvert_1}{\alpha''}+1}}{\alpha''}\|\phi\|_{L^\infty}\llvert\phi\rrvert_1 \|B\|_{\alpha''}\\
&\qquad\times\left\{\sup_{\theta^\pm\in L^1}\left(\llvert\theta^+\rrvert_1 \exp\left(-\Bigl(\frac{1}{\alpha'}-\frac{1}{\alpha''}\Bigr)(\llvert\theta^+\rrvert_1 +\llvert\theta^-\rrvert_1 )\right)\right)\right.\\
&\left.\qquad\qquad+ \sup_{\theta^\pm\in L^1}\left(\llvert\theta^-\rrvert_1 \exp\left(-\Bigl(\frac{1}{\alpha'}-\frac{1}{\alpha''}\Bigr)(\llvert\theta^+\rrvert_1 +\llvert\theta^-\rrvert_1 )\right)\right)\right\},
\end{align*}
and the proof follows using the inequalities
$xye^{-n(x+y)}=(xe^{-nx})(ye^{-ny})\leq\frac{1}{e^2n^2}$ and
$xe^{-n(x+y)}\leq xe^{-nx}\leq\frac{1}{en}$ for $x,y\geq0$, $n>0$.
\end{proof}

\begin{theorem}\label{Theorem2}
Given an $0<\alpha<\alpha_0$, let
$B_{t,\mathrm{ren}}^{(\eps)}, B_{t,LP}$,
$t\in\left[0,T\right)$, be the local solutions in $\mathcal{E}_{\alpha}$ to
the initial value problems \eqref{V12}, \eqref{V19}
with $B_{0,\mathrm{ren}}^{(\eps)},B_{0,LP}\in\mathcal{E}_{\alpha_0}$. If
$0\leq\phi\in L^1\cap L^\infty$ and $\lim_{\eps\searrow 0}\|B_{0,\mathrm{ren}}^{(\eps)}-B_{0,LP}\|_{\alpha_0}=0$, then, for each $t\in\left[0,T\right)$,
\[
\lim_{\eps\searrow 0}\|B_{t,\mathrm{ren}}^{(\eps)}-B_{t,LP}\|_{\alpha}=0.
\]
\end{theorem}

\begin{proof}
This result follows as a consequence of Proposition \ref{Proposition5} and a
particular application of Theorem \ref{Thconv} for $p=2$ and
\[
N_\eps=2\eps^d z\|\phi\|_{L^\infty}e^{\frac{\llvert\phi\rrvert_1}{\alpha}}\max\left\{
\alpha_0\llvert\phi\rrvert_1,\frac{\alpha_0^3}{e}\right\}.\qedhere
\]
\end{proof}

A purpose of considering a mesoscopic limit of a given interacting particle
system is to derive a kinetic equation which in a closed form describes a
reduced system in such a way that it reflects some properties of the initial
one. To do this one should prove that the derived limiting time evolution
satisfies
the so-called chaos propagation principle. Namely, if one considers as an
initial distribution a Poisson product measure $\pi_{\rho^+_0dx,\rho^-_0dx}^2=\pi_{\rho_0^+ dx}\otimes\pi_{\rho_0^-dx}$, $\rho^\pm_0>0$, then, at each moment of
time $t>0$, the distribution must be Poissonian as well. Observe that due to
\eqref{BF_via_cf} and \eqref{meanLP}, the GF corresponding to a Poisson
product measure has an exponential form. This leads to the choice of an
initial GF in \eqref{V19}.

\begin{theorem}\label{reducedtokinetic}
If the initial condition $B_{0,LP}$ in \eqref{V19} is of the type
\[
B_{0,LP}(\theta^+,\theta^-)=\exp\left(\int_{\mathbb{R}^d}dx\,\rho^+_0(x)\theta^+(x)+
\int_{\mathbb{R}^d}dy\,\rho^-_0(y)\theta^-(y)\right),\quad\theta^\pm\in L^1
\]
for some $\rho_0^+,\rho_0^-\in L^\infty$ such that
$\|\rho_0^\pm\|_{L^\infty}\leq \frac{1}{\alpha_0}$, then the functional defined
for all $\theta^\pm\in L^1$ by
\begin{equation}\label{expsol}
B_{t,LP}(\theta^+,\theta^-)=\exp\left(\int_{\mathbb{R}^d}dx\,
\rho^+_t(x)\theta^+(x)+\int_{\mathbb{R}^d}dy\,\rho^-_t(y)\theta^-(y)\right),
\end{equation}
solves the initial value problem \eqref{V19} for $t\in\left[0,T\right)$,
provided $\rho^+_t,\rho^-_t$ are classical solutions to the system of equations
\begin{equation}\label{kinetic}
\begin{cases}
\dfrac{\partial}{\partial t}\rho_t^+=-m\rho_t^++ze^{-(\rho^-_t*\phi)},\\[6mm]
\dfrac{\partial}{\partial t}\rho_t^-=-m\rho_t^-+ze^{-(\rho^+_t*\phi)},
\end{cases} \quad t\in\left[0,T\right), x\in\R^d,
\end{equation}
such that, for each $t\in\left[0,T\right)$, $\rho^+_t,\rho^-_t\in L^\infty$ and
$\|\rho_t^\pm\|_{L^\infty}\leq \frac{1}{\alpha}$. Here $*$ denotes the usual
convolution of functions,
\[
(\rho_t^\pm * \phi)(x):=\int_{\R^d}\, dy\,  \phi(x-y)\rho_t^\pm(y), \quad x\in\R^d.
\]
\end{theorem}

\begin{proof}
Let $B_{t,LP}$ be given by \eqref{expsol}. Then, for any
$\theta^\pm,\theta_1^+\in L^1$ one has
\[
\frac{\partial}{\partial z}
B_{t,LP}(\theta^++z\theta_1^+,\theta^-)\biggr|_{z=0}=B_{t,LP}(\theta^+,\theta^-)\int_{\mathbb{R}^d}dx
\rho_t^+(x)\theta_1^+(x),
\]
meaning $\vd B_{t,LP}(\theta^+,\theta^-;x,\emptyset)= B_{t,LP}(\theta^+,\theta^-) \rho_t^+(x)$. In a similar way one can show that
$\vd B_{t,LP}(\theta^+,\theta^-;\emptyset,y)= B_{t,LP}(\theta^+,\theta^-) \rho_t^-(y)$. Hence, for all $\theta^\pm\in L^1$,
\begin{align*}
&(\widetilde L_{LP} B_{t,LP})(\theta^+,\theta^-)\\
=&-B_{t,LP}(\theta^+,\theta^-)\int_{\R^d}dx\,\theta^+(x)\left(m\rho^+_t(x) - z
\exp\left(-(\rho_t^-*\phi)(x))\right)\right)\\
&-B_{t,LP}(\theta^+,\theta^-)\int_{\R^d}dy\,\theta^-(y)\left(m\rho^-_t(y) - z
\exp\left(-(\rho_t^+*\phi)(y))\right)\right).
\end{align*}
That is, if $\rho^\pm_t$ are classic solutions to \eqref{kinetic},
then the right-hand side of the latter equality is equal to
\[
B_{t,LP}(\theta^+,\theta^-)\frac{d}{dt}\left\{\int_{\mathbb{R}^d}dx\,
\rho_t^+(x)\theta^+(x) + \int_{\mathbb{R}^d}dy\,\rho_t^-(y)\theta^-(y)
\right\}= \frac{\partial}{\partial t} B_{t,LP}(\theta^+,\theta^-).
\]
This proves that $B_{t,LP}$, given by \eqref{expsol}, solves equation
\eqref{V19}. If, in addition, $\rho^\pm_t\in L^\infty$ with
$\|\rho_t^\pm\|_{L^\infty}\leq \frac{1}{\alpha}$, then one concludes from
\eqref{norminscale} that $B_{t,LP}\in\mathcal{E}_\alpha$ (the entireness of
$B_{t,LP}$ is clear by its definition \eqref{expsol}). The uniqueness of the
solution to \eqref{V19} completes the proof.
\end{proof}

Observe that the statement of Theorem \ref{reducedtokinetic} does not consider
any positiveness assumption on $\rho_t^\pm$. However, having in mind the
propagation of the chaos property, we are mostly interested in positive
solutions to the system \eqref{kinetic}. The next theorem states conditions
for the existence and uniqueness of such solutions.

\begin{theorem}\label{Theorem3}
Let $0\leq \rho_0^\pm\in L^\infty(\R^d)$ be given and let $c_0>0$ be such that
$\lVert \rho_0^\pm\rVert_{L^\infty}\leq c_0$. Set
$c=\max\bigl\{ c_0,\frac{z}{m}\bigr\}$. Then there exists a solution to
\eqref{kinetic} such that $0\leq \rho_t^\pm\in L^\infty$, $t>0$, and
\begin{equation}\label{eest}
\lVert \rho_t^\pm\rVert_{L^\infty}\leq c, \quad t>0.
\end{equation}
Such a solution is the unique non-negative solution to \eqref{kinetic} which
fulfills \eqref{eest}.
\end{theorem}

\begin{proof}
For $T>0$ fixed, let us consider the Banach space $L^\infty \times L^\infty$
with the norm
\[
\lVert (v^+,v^-)\rVert_\infty:=\lVert v^+\rVert_{L^\infty}+\lVert v^-\rVert_{L^\infty}
\]
and the Banach space of all $L^\infty \times L^\infty$-valued continuous
functions on $[0,T]$,
\begin{equation*}
X_T:=C\bigl([0,T]\rightarrow L^\infty \times L^\infty\bigr),
\end{equation*}
with the norm defined for all $v\in X_T$,
$v:\ [0,T]\ni t\mapsto v_t=(v^+_t,v_t^-)\in L^\infty\times L^\infty$, by
\[
\lVert v\rVert_T:=\max_{t\in[0,T]}\lVert (v_t^+,v_t^-)\rVert_\infty.
\]
Let $X_T^+$ be the cone of all elements $v\in X_T$ such that, for all
$t\in[0,T]$, $v_t^\pm(x)\geq0$ for a.a.~$x\in\R^d$. For an arbitrary $c>0$, we
denote by $B_{T,c}^+$ the intersection of the cone $X_T^+$ with the closed ball
$B_{T,c}:=\{v\in X_T: \lVert v\rVert_T \leq 2c\}$.

Given a $0\leq \rho_0^\pm\in L^\infty(\R^d)$ such that
$\lVert \rho_0^\pm\rVert_{L^\infty}\leq c_0$ for some $c_0>0$, let $\Phi$ be
the mapping which assigns, for each $v=(v^+,v^-)\in B_{T,c}^+$, the solution
$u:=(u^+,u^-)$ to the system of linear non-homogeneous
equations
\begin{equation}\label{linearkinetic}
\begin{cases}
\dfrac{\partial}{\partial t}u_t^+(x)=-mu_t^+(x)+ze^{-(v^-_t*\phi)(x)},\\[6mm]
\dfrac{\partial}{\partial t}u_t^-(x)=-mu_t^-(x)+ze^{-(v^+_t*\phi)(x)},
\end{cases} \quad t\in\left[0,T\right], \ \textrm{a.a.}\ x\in\R^d,
\end{equation}
for the initial conditions ${u_t^\pm}_{|t=0}=\rho_0^\pm$. That is,
$u=\Phi v:=((\Phi v)^+,(\Phi v)^-)$. Actually, straightforwardly calculations
show that, for each $v=(v^+,v^-)\in B_{T,c}^+$, $\Phi v=((\Phi v)^+,(\Phi v)^-)$
is explicitly given for all $t\in [0,T]$ and a.a.~$x\in\R^d$ by
\begin{equation*}
 (\Phi v)_t^\pm(x)=e^{-mt}\rho_0^\pm(x)+z\int_0^t ds\, e^{-m(t-s)}e^{-(v^\mp_s*\phi)(x)}\geq 0.
\end{equation*}
Moreover, by the positiveness assumptions on $v^\pm$ and $\phi$, one finds
\[
\lVert (\Phi v)_t^\pm \rVert_{L^\infty} \leq c_0e^{-mt}+z\int_0^tds\, e^{-m(t-s)}=c_0e^{-mt}+\frac{z}{m}(1-e^{-mt})\leq c,\quad t\in\left[0,T\right],
\]
where $c:=\max\bigl\{ c_0,\frac{z}{m}\bigr\}$, showing that
$\Phi v \in B_{T,c}^+$ for all $v\in B_{T,c}^+$.

For all $v,w\in B_{T,c}^+$ and all $t\in\left[0,T\right]$ one has
\[
\lVert (\Phi v)_t - (\Phi w)_t\rVert_\infty =\lVert (\Phi v)_t^+ - (\Phi w)_t^+\rVert_{L^\infty}
+\lVert (\Phi v)^-_t - (\Phi w)^-_t\rVert_{L^\infty}
\]
with
\begin{align*}
\bigl\lvert (\Phi v)^\pm_t(x) - (\Phi w)^\pm_t(x)\bigr\rvert
&\leq  z\int_0^tds\, e^{-m(t-s)}\Bigl\lvert e^{-(v_s^\mp \ast \phi)(x)}-e^{-(w_s^\mp \ast \phi)(x)}\Bigr\rvert \\
&\leq z\llvert\phi\rrvert_1 \sup_{s\in[0,t]} \lVert v_s^\mp-w_s^\mp\rVert_{L^\infty}\frac{1-e^{-mt}}{m},
\end{align*}
where in the latter inequality we have used the inequalities
$|e^{-a}-e^{-b}|\leq |a-b|$, $a,b\geq 0$ and $\|f*g\|_{L^\infty}\leq \llvert f\rrvert_1\|g\|_{L^\infty}$, $f\in L^1$, $g\in L^\infty$. Therefore, for any
$t\in[0,T]$,
\[
\lVert (\Phi v)_t^\pm - (\Phi w)_t^\pm\rVert_{L^\infty}\leq z \llvert\phi\rrvert_1 T \sup_{s\in[0,T]} \lVert v_s^\mp-w_s^\mp\rVert_{L^\infty},
\]
and thus
\[
\lVert \Phi v -\Phi w\rVert_T \leq z \llvert\phi\rrvert_1 T \lVert  v -w\rVert_T.
\]
As a consequence, the mapping $\Phi$ is a contraction on the metric space
$B_{T,c}^+$ whenever $T<\frac{1}{z \llvert\phi\rrvert_1}$. In such a situation,
there is a unique fixed point $\rho=(\rho^+,\rho^-)\in B_{T,c}^+$, i.e.,
$\Phi \rho=\rho$, which leads to a unique solution to the system of equations
\eqref{kinetic} on the interval $[0,T]$.

Now let us consider \eqref{kinetic}, \eqref{linearkinetic} on the time
interval $[T,2T]$ with the initial condition given by $\rho_T$. By the
previous construction, $\lVert\rho_T^\pm\rVert_{L^\infty}\leq c$. One can then
repeat the above arguments in the same metric space $B_{T,c}^+$, because
$\max\bigl\{ c,\frac{z}{m}\bigr\}=c$ and, for any $t\in [T,2T]$,
\[
\int_T^t ds\, e^{-m(t-s)}=\frac{1-e^{-m(t-T)}}{m}\leq t-T\leq T.
\]
This argument iterated for the intervals $[2T,3T]$, $[3T,4T]$, etc, yields at
the end the complete proof of the required result.
\end{proof}

\section{Equilibrium: multi-phases and stability}\label{Section5}

In this section we realize the analysis of the system of kinetic equations
\eqref{kinetic} in the space-homogeneous case. More precisely, we consider the
stationary system corresponding to the space-homogeneous version of
\eqref{kinetic},
\begin{equation}
\begin{cases}
-m\rho _{t}^{+}+ze^{-\beta \rho _{t}^{-}}=0, \\
-m\rho _{t}^{-}+ze^{-\beta \rho _{t}^{+}}=0,
\end{cases}
\end{equation}%
where
\[
\beta:=\int_{\R^d} dx\, \phi (x) >0.
\]
Observe that for $r_{t}^{\pm }=\beta \rho _{t}^{\pm }$,
$a=\dfrac{z}{m}\beta >0$ one obtains the following system
\begin{equation}
\left\{
\begin{aligned}
ae^{-r^{-}}=&r^{+} \\
ae^{-r^{+}}=&r^{-}
\end{aligned}%
\right.   \label{stat-eqn-mod}
\end{equation}%
and thus
\[
r^{\pm} =a\exp \left( -a\exp \left( -r^{\pm}\right) \right).
\]

\begin{proposition}\label{propsda}
Given an $a>0$, let $f$ be the function defined on $\left[0,+\infty\right[$  by
\begin{equation*}
f\left( x\right) =a\exp \left( -a\exp \left( -x\right) \right)-x,\quad x\geq 0.
\end{equation*}
If $a\leq e$, then there is a unique positive root $x_{0}$ of $f$. Moreover,
$x_{0}=a\exp \left( -x_{0}\right) $. If $a>e$, then
there are three and only three positive roots $x_{1}<x_{2}<x_{3}$ of $f$.
Moreover, $x_{1}=a\exp
\left( -x_{3}\right) $, $x_{2}=a\exp \left( -x_{2}\right) $, $%
x_{3}=a\exp \left( -x_{1}\right) $ and%
\begin{align}
0 <&x_{1}<a\exp \left( -\frac{a}{e}\right),  \label{smallfirstroot} \\
a >&x_{3}>a\exp \left( -a\exp \left( -\frac{a}{e}\right) \right) .
\label{bigthirdroot}
\end{align}
\end{proposition}

\begin{proof}
First of all, let us observe that if $f\left( x\right) =0$, then $a\exp \left(
-x\right) =-\ln \frac{x}{a}$, which means that $\ln \frac{x}{a}<0$ and thus
\begin{equation}
x<a.  \label{x_less_a}
\end{equation}
Furthermore, $-\ln \frac{x}{a}$ is also a root of $f$:
\begin{equation*}
f\left( -\ln \frac{x}{a}\right)  =a\exp \left( -a\exp \left( \ln
\frac{x}{a}\right)
\right) +\ln \frac{x}{a}
=a\exp \left( -x\right) +\ln \frac{x}{a}=0.
\end{equation*}

Let us consider
\[
f^{\prime }\left( x\right) =a^{2}\exp \left( -a\exp \left( -x\right)
\right) \exp \left( -x\right) -1.
\]
Using the well-known inequality $te^{-t}\leq e^{-1}$, $t\geq 0$, with
$t=ae^{-x}$, $x\geq 0$, we obtain
\[
f^{\prime }\left( x\right) \leq \frac{a}{e}-1.
\]

Therefore, if $a< e$, $f$ is a strictly decreasing function on $[0,+\infty )$.
For $a=e$, we have $f^\prime(x)\leq 0$ for all $x\geq 0$ with $f^\prime(x)=0$
only for $x=1$. Independently of the case under consideration, in addition,
one has $f\left( 0\right) =ae^{-a}>0$ and
$\lim_{x\rightarrow+\infty}f(x)=-\infty$, which implies that $f$ has only one
positive root $x_{0}$. Due to the initial considerations, then
$x_{0}=-\ln \frac{x_{0}}{a}$, that is, $x_{0}=a\exp\left( -x_{0}\right) $.

Let now $a>e$. Since $f^{\prime }\left( x\right) =0$ implies
\[
-a\exp \left( -x\right) -x=-\ln a^{2},
\]%
let us consider the following auxiliary function
\[
g\left( x\right) =x+a\exp \left( -x\right) -2\ln a,\quad x\geq 0,
\]
which allows to rewrite $f^{\prime }$ as
\[
f^{\prime }\left( x\right) =a^{2}\exp \left( -g\left( x\right) -2\ln
a\right) -1=\exp \left( -g\left( x\right) \right) -1.
\]%
Concerning the function $g^\prime$,
\[
g^{\prime }\left( x\right) =1-a\exp \left( -x\right),\quad x\geq 0,
\]
one has $g^{\prime }\left( x\right) =0$ only for $x=\ln a$. For
$x>\ln a$ we have $g^{\prime }\left( x\right)>0$, meaning that $g$ is
increasing on $[\ln a,+\infty)$. Since the sign of the $g^{\prime }$ is the
same in whole the interval $[0,\ln a)$, in particular, it coincides with the
sign of $g^{\prime }\left( 0\right) =1-a<0$. Thus, $g$ is strictly decreasing
on $[0,\ln a)$. As a result, on $[0,+\infty )$ the function $g$ has a unique
minimum. Moreover,  $\lim_{x\rightarrow +\infty }g\left( x\right) =+\infty $,
\[
g\left( \ln a\right) =\ln a+1-2\ln a=1-\ln a<0,
\]
and $g\left(0\right) =a-2\ln a>0$, which follows from the fact that for the
function $h\left( t\right)=t-2\ln t$ one has
$h^{\prime }\left( t\right) =1-\frac{2}{t}=\frac{t-2}{t}>0$, $t>e$,
and thus $h\left( a\right) >h\left( e\right) =e-2>0$. Consequently, $g$ has two
positive roots, say $y_1$, $y_2$, $0<y_{1}<y_{2}<+\infty $. In terms of the
function $f$, this implies that $f^{\prime }>0$ on
$\left( y_{1},y_{2}\right) $ (where $g<0$) and $f^{\prime }<0$ on
$[0,y_{1})\cup\left( y_{2},+\infty \right) $, meaning that $y_{1}$ is the point
of the minimum of the function $f$ and $y_{2}$ is the point of the maximum of
$f$.

The number of positive roots of $f$ depends on the sign of $f\left(y_{j}\right) $, $j=1,2$. Let us prove that
\begin{equation}
f\left( y_{1}\right) <0<f\left( y_{2}\right),  \label{main}
\end{equation}%
which then implies that the function $f$ has three and only three roots.

As $g\left( y_{j}\right)=0$, $j=1,2$, which implies that
$-a\exp \left( -y_{j}\right) =y_{j}-2\ln a$, one has
\[
f\left( y_{j}\right)  =a\exp \left( y_{j}-2\ln a\right) -y_{j}
=\frac{1}{a}e^{y_{j}}-y_{j} =e^{y_{j}}\left(
\frac{1}{a}-y_{j}e^{-y_{j}}\right).
\]
Let us consider the function $p\left( t\right) =\frac{1}{a}-te^{-t}$,
$t\geq 0$. Since $p^{\prime }\left( t\right) =\left( t-1\right) e^{-t}$, this
function has a minimum at the point $t=1$,
$p\left( 1\right) =\frac{1}{a}-\frac{1}{e}<0$. Moreover,
$p\left( 0\right) =\frac{1}{a}>0$ and
$\lim_{t\rightarrow +\infty }=\frac{1}{a}>0$. Therefore, this
function has two roots, $0<t_{1}<1<t_{2}<+\infty $, and $p<0$ on $\left( t\,_{1},t_{2}\right) $, $p>0$ on $[0,t_{1})\cup \left( t_{2},+\infty \right)$. Thus,
inequality \eqref{main} will follow from the inequality
\begin{equation}
t_{1}<y_{1}<t_{2}<y_{2}.\label{Janeiro}
\end{equation}
In order to show \eqref{Janeiro}, first we observe that $t_1<t_2$ are the only
roots of $p$. However, for $y_{1}<y_{2}$ one finds $p\left(a\exp \left( -y_{j}\right)\right)=\frac{1}{a}f^\prime(y_j)=0$, $j=1,2$, meaning that
\[
a\exp \left( -y_{2}\right) =t_{1}<t_{2}=a\exp \left( -y_{1}\right) .
\]
Hence, to prove the sequence of inequalities \eqref{Janeiro} is
equivalent to show
\[
e^{-t_{1}} >e^{-y_{1}}>e^{-t_{2}}>e^{-y_{2}}\Longleftrightarrow
e^{-t_{1}} >\frac{t_{2}}{a}>e^{-t_{2}}>\frac{t_{1}}{a}
\]
or
\[
\frac{1}{at_{1}} >\frac{t_{2}}{a}>\frac{1}{at_{2}}>\frac{t_{1}}{a}.
\]
Since $t_{2}>1$, the latter three inequalities hold if and only if
\begin{equation}
t_{1}t_{2}<1.  \label{rootsineq}
\end{equation}%
So we will prove \eqref{rootsineq}. Since $w(t)=te^{-t}$ is a increasing
function on $\left[0,1\right)$ and $t_2>1$ ($\frac{1}{t_2}<1$), observe that
to show \eqref{rootsineq} it is enough to prove
\begin{equation}
\frac{1}{t_{2}}\exp \left( -\frac{1}{t_{2}}\right) >\frac{1}{a}=t_{2}\exp
\left( -t_{2}\right),\label{Fevereiro}
\end{equation}
because due to the fact that $p(t_j)=0$, $j=1,2$, the right-hand side of
\eqref{Fevereiro} is also equal to $t_{1}\exp\left( -t_{1}\right)$. Concerning
\eqref{Fevereiro}, note also that it is equivalent to
\begin{eqnarray}
t_{2}^{2} <\exp \left( t_{2}-\frac{1}{t_{2}}\right)&\Longleftrightarrow&
\exp \left( 2t_{2}-2\ln a\right)  <\exp \left( t_{2}-\frac{1}{t_{2}}\right)\nonumber\\
&\Longleftrightarrow& t_{2}+\frac{1}{t_{2}} <2\ln a\nonumber\\
&\Longleftrightarrow&t_{2}^{2}-\left( 2\ln a\right) t_{2}+1 <0.\label{dopineq1}
\end{eqnarray}
Clearly, the solutions to the inequality $v(t)=t^2-(2\ln a)t+1<0$ are
$t\in(\ln a-\sqrt{\ln ^{2}a-1},\ln a+\sqrt{\ln ^{2}a-1})$. Since
$v(1)=2(1-\ln a)<0$ and $t_2>1$, inequality \eqref{dopineq1} holds if and only
if
\begin{equation}
t_{2} <\ln a+\sqrt{\ln ^{2}a-1}.\label{dopineq2}
\end{equation}
In addition, because $w(t)=te^{-t}$ is a decreasing function for $t>1$ and
$w(t_2)=\frac{1}{a}$, inequality \eqref{dopineq2} holds if and only if
\begin{eqnarray*}
&&w\left(\ln a+\sqrt{\ln ^{2}a-1}\right)<w(t_2)\\
&&\Longleftrightarrow\left( \ln a+\sqrt{\ln ^{2}a-1}\right) \exp \left( -\ln a-\sqrt{\ln ^{2}a-1}\right)  <\frac{1}{a},
\end{eqnarray*}
which is equivalent to
\begin{eqnarray*}
&&\left( \ln a+\sqrt{\ln ^{2}a-1}\right) \exp \left( -\sqrt{\ln ^{2}a-1}
\right)  <1\\
&&\Longleftrightarrow \ln a+\sqrt{\ln ^{2}a-1} <\exp \left( \sqrt{\ln ^{2}a-1}\right) .
\end{eqnarray*}
Set
\[
u\left( y\right) =e^{y}-y-\sqrt{y^{2}+1},\quad y\geq 0.
\]
We have
\begin{align*}
u^{\prime }\left( y\right)  =&e^{y}-1-\frac{y}{\sqrt{y^{2}+1}}, \\
u^{\prime \prime }\left( y\right)  =&e^{y}-\frac{\sqrt{y^{2}+1}-\frac{y^{2}%
}{\sqrt{y^{2}+1}}}{y^{2}+1} =e^{y}-\frac{1}{\left( y^{2}+1\right) ^{\frac{3}{2}}},
\end{align*}%
with $e^{y}\geq 1$ and $\frac{1}{\left( y^{2}+1\right) ^{\frac{3}{2}}}\leq 1$.
Therefore, $u^{\prime \prime }\geq 0$ and $u^{\prime \prime }\left(
y\right) =0$ only for $y=0$, meaning that $u^{\prime }$ is a increasing
function. Hence,
$u^{\prime }\left( y\right) \geq u^{\prime }\left( 0\right) =0$ for all
$y\geq 0$. We have $u^{\prime }\left( y\right) =0$ only for $y=0$. Therefore,
also $u$ is increasing, and thus $u\left( y\right) >u\left( 0\right) =0$ for
all $y>0$. In particular, for $y=\sqrt{\ln ^{2}a-1}>0$.

As result, for $a>e$ there are three and only three positive roots of $f$, say
$x_{1}<x_{2}<x_{3}$.\footnote{Of course, $x_1<y_1<x_2<y_2<x_3$.} By the
considerations at the beginning,
$-\ln\frac{x_{3}}{a}<-\ln\frac{x_{2}}{a}<-\ln \frac{x_{1}}{a}$ are also
positive roots of $f$. Hence,
\[
x_{1}=-\ln \frac{x_{3}}{a},~~x_{2}=-\ln \frac{x_{2}}{a},~~x_{3}=-\ln \frac{%
x_{1}}{a},
\]
that is,
\[
x_{3}=a\exp \left( -x_{1}\right) ,~~x_{2}=a\exp \left( -x_{2}\right)
,~~x_{1}=a\exp \left( -x_{3}\right) .
\]%
To prove \eqref{smallfirstroot} and \eqref{bigthirdroot} we recall that
$x_{1}<y_{1}<\ln a$, where the latter inequality follows from the fact that
the function $g$ is decreasing on $[0,\ln a]$ with $g(\ln a)<0=g(y_1)$.
Therefore,
\[
x_{3}=a\exp \left( -x_{1}\right) >a\exp \left( -\ln a\right) =1.
\]
Moreover, since $w(t)=te^{-t}$ is decreasing on $[1,+\infty[$ and thus
\[
f\left( 1\right)  =a\exp \left( -\frac{a}{e}\right)
-1=e\frac{a}{e}\exp
\left( -\frac{a}{e}\right) -1 <e\frac{1}{e}-1=0,
\]
one may conclude that $x_{1}<1<x_{2}$. Hence, $x_{3}=a\exp \left(
-x_{1}\right) >\frac{a}{e}$. Then, finally,
\[
0<x_{1}=a\exp \left( -x_{3}\right) <a\exp \left( -\frac{a}{e}\right)
\]
and, by \eqref{x_less_a},
\[
a>x_{3}=a\exp \left( -x_{1}\right) >a\exp \left( -a\exp \left( -\frac{a}{e}\right) \right) .
\]
The statement is fully proven.
\end{proof}

\begin{theorem}\label{th-homog-kinetic}
Consider the space-homogeneous version of the system of equations
\eqref{kinetic}
\begin{equation}\label{homog-kinetic}
\begin{cases}
\dfrac{d}{dt}\rho_t^+=-m\rho_t^++ze^{-\beta \rho^-_t},\\[6mm]
\dfrac{d}{d t}\rho_t^-=-m\rho_t^-+ze^{-\beta\rho^+_t},
\end{cases} \quad t\in\left[0,T\right),
\end{equation}
where $\beta:=\int_{\R^d} dx\, \phi (x)$, $a:=\dfrac{z}{m}\beta$. Let
$x_0,x_1,x_2,x_3$ be the positive roots given by Proposition~\ref{propsda}.
If $a \leq e$, then there is a unique equilibrium solution $\left( \frac{1}{\beta}x_{0},\frac{1}{\beta}x_{0}\right) $ to \eqref{homog-kinetic}. For $a<e$,
this solution is a stable node, while for $a=e$ it is a saddle-node
equilibrium point. If $a>e$, then there are three and only three equilibrium
solutions $\left( \frac{1}{\beta}x_{1},\frac{1}{\beta}x_{3}\right) $, $\left( \frac{1}{\beta}x_{2},\frac{1}{\beta}x_{2}\right) $, $\left(\frac{1}{\beta} x_{3},\frac{1}{\beta}x_{1}\right) $ to \eqref{homog-kinetic}. The second solution is
a saddle point and the other two solutions are stable nodes of \eqref{homog-kinetic}.
\end{theorem}

\begin{proof}
First of all note that properties of stationary points of \eqref%
{homog-kinetic} are the same as the corresponding properties for the system of
equations
\begin{equation}\label{equivsystem}
\begin{cases}
\dfrac{d}{dt}r_{t}^{+}=P\left( r_{t}^{+},r_{t}^{-}\right)  \\[3mm]
\dfrac{d}{dt}r_{t}^{-}=Q\left( r_{t}^{+},r_{t}^{-}\right)
\end{cases}
\end{equation}
where $r_{t}^{\pm }=\beta \rho _{t}^{\pm }$ and
\[
P\left( x,y\right) =-mx+mae^{-y},\qquad Q\left( x,y\right) =-my+mae^{-x}.
\]
Clearly, equilibrium points of \eqref{equivsystem} do not depend on $m$. They
solve \eqref{stat-eqn-mod} and can be obtained from Proposition~\ref{propsda}.

To study the character of the equilibrium points of \eqref{equivsystem}, let
us consider the following matrix
\[
A \left( x,y\right)  :=\left(
\begin{array}{cc}
\dfrac{\partial P}{\partial x}  & \dfrac{\partial P}{\partial y%
} \vphantom{\dfrac{a}{\dfrac{a}{a}}}\\
\dfrac{\partial Q}{\partial x}  & \dfrac{\partial Q}{\partial y%
}
\end{array}\right)= \left(
\begin{array}{cc}
-m & -mae^{-y} \\
-mae^{-x} & -m
\end{array}
\right)
\]
We have
\begin{align}
D(x,y)&:=\det A(x,y) =m^2-a^{2}m^2e^{-x}e^{-y},\label{detne}\\
T(x,y)&:=\mathrm{tr}\, A(x,y)=-2m<0,\notag
\end{align}
and
\[
T ^{2}\left( x,y\right) -4 D \left( x,y\right)
=4m^2a^{2}e^{-x}e^{-y}>0.
\]
Therefore, by e.g.~\cite{HSD2004}, the nature of an equilibrium point of
\eqref{equivsystem} (and thus of \eqref{homog-kinetic}) depends on the sign
of $D(x,y)$ at that point.

For $a< e$, one has from \eqref{detne} that $D(x_0,x_0)>0$ if and only if
$e^{x_0}>a$, which is equivalent to $x_0e^{x_0}>ax_0$ and to $a>ax_0$, $x_0<1$,
where we have used the equality $x_0e^{x_0}=a$ given by Proposition \ref{propsda}.
The latter inequality is true, since the function $h(x)=xe^x$ is strictly
increasing and the equation $xe^x=1$ has a unique solution, $x=1$. Therefore,
a solution to $xe^x=a< e$ should be strictly smaller than $1$. Hence, $
(x_0,x_0)$ is a stable node of \eqref{equivsystem}.

Similarly, for $a>e$, \eqref{detne} yields $D(x_2,x_2)<0$ if and only if
$e^{x_2}<a$, which holds because $x_2>1$ and $x_2e^{x_2}=a$, cf.~Proposition
\ref{propsda}. Hence, $(x_2,x_2)$ is a saddle point of \eqref{equivsystem}.

Still for the case $a>e$, $D(x_1,x_3)=D(x_3,x_1)>0$ if and only if
$e^{x_1+x_3}>a^2$. Since $x_1=ae^{-x_3}$ and $x_3=ae^{-x_1}$ (Proposition
\ref{propsda}), the latter inequality is equivalent to $x_1x_3<1$. To show
that $x_1x_3<1$, let us consider the function $r(t)=ate^{-t}$, which is
strictly decreasing for $t>1$. Since the equation $r(t)=1$ is equivalent to
$p(t)=0$, where $p$ is the function defined in the proof of Proposition
\ref{propsda}, the solutions to $r(t)=1$ are the roots of $p$, that is, $t_1$,
$t_2$. Therefore, it follows from the proof of Proposition \ref{propsda} that
$1<t_2<x_3$, leading to
\[
1=r(t_2)>r(x_3)=ax_3e^{-x_3}=x_1x_3.
\]
Hence, $(x_1,x_3)$ and $(x_3,x_1)$ are also stable nodes of \eqref{equivsystem}.

Finally, for $a=e$, one has $x_0=ae^{-x_0}=e^{1-x_0}$, and thus $x_0=1$.
Therefore, $D(x_0,x_0)=D(1,1)=0$ and one has a saddle-node equilibrium point.
\end{proof}

As a result, one has a bifurcation in the system \eqref{homog-kinetic}
depending on the value of $a=\frac{z}{m}\beta$.

In Appendix below, we present numerical solutions to \eqref{homog-kinetic} for
different values of $a$. Namely, we consider a set of initial values
$\rho_0^\pm$ from the interval $[0,2]$ with step $0.5$ and we draw the
corresponding graphs of, say, $\rho_t^+$ on the time interval $t\in[0,200]$.
Of course, the graphs of $\rho_t^-$ have the same shape. As one can see in
Figure~\ref{fig:test}, there is a unique stable solution for $a<e$ (that is,
$\frac{x_0}{\beta}$). For $a>e$, one has two stable solutions
($\frac{x_1}{\beta}$ and $\frac{x_3}{\beta}$). For $a=e$, stable
solutions do not exist at all. The corresponding phase plane pictures are
presented in Figure~\ref{fig:test2}.

\section*{Acknowledgments}
Financial support of DFG through CRC 701, Research Group ``Stochastic
Dynamics: Mathematical Theory and Applications'' at ZiF, and FCT through
PTDC/MAT/100983/2008, PTDC/MAT-STA/1284/2012 and PEst OE/MAT/UI0209/2013 are
gratefully acknowledged.

\appendix

\section{Appendix}

\begin{figure}[ht]
\begin{subfigure}{.5\linewidth}
\centering
\includegraphics[scale=.35]{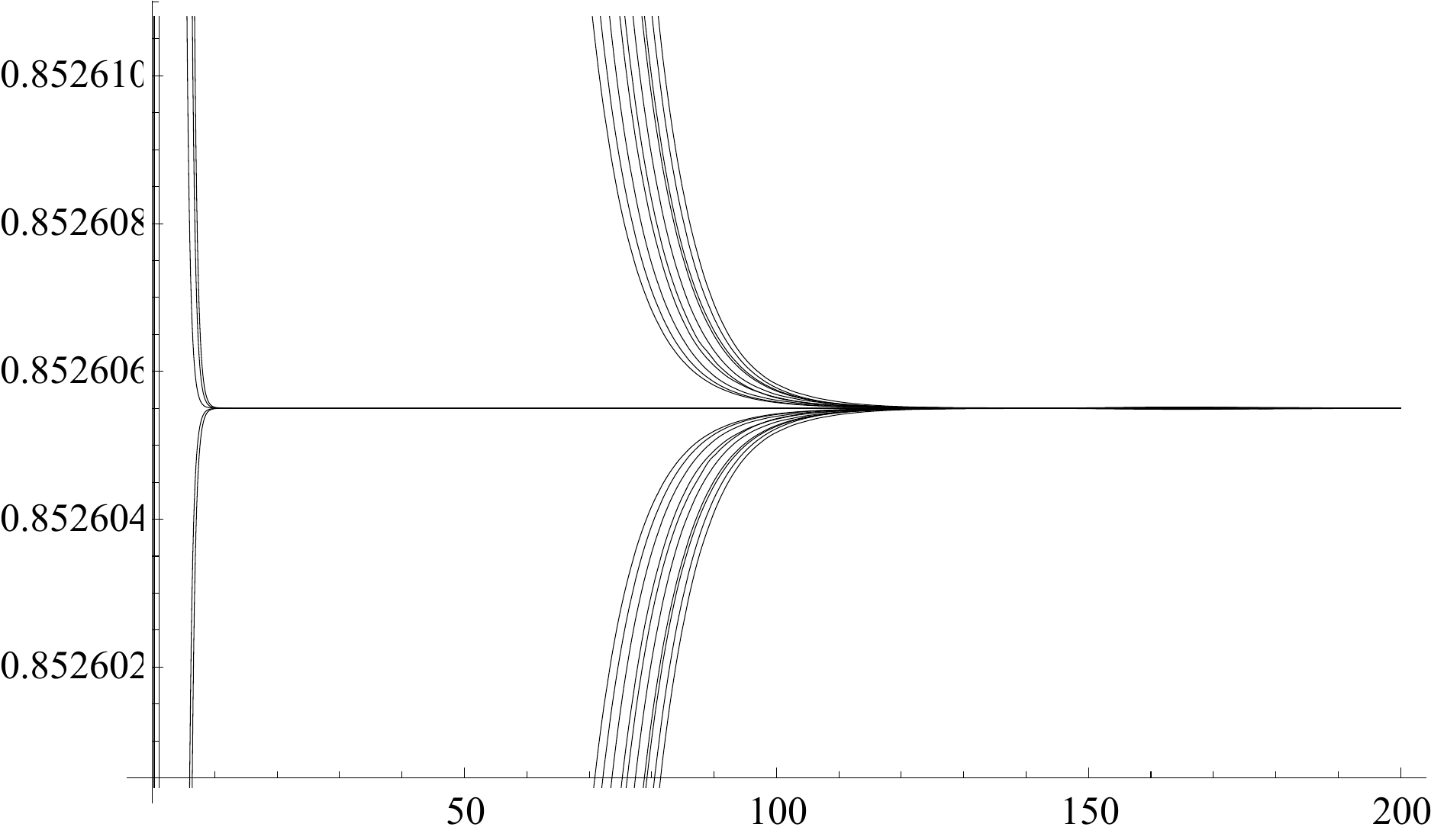}
\caption{$a=2$}
\label{fig:sub1}
\end{subfigure}%
\begin{subfigure}{.5\linewidth}
\centering
\includegraphics[scale=.35]{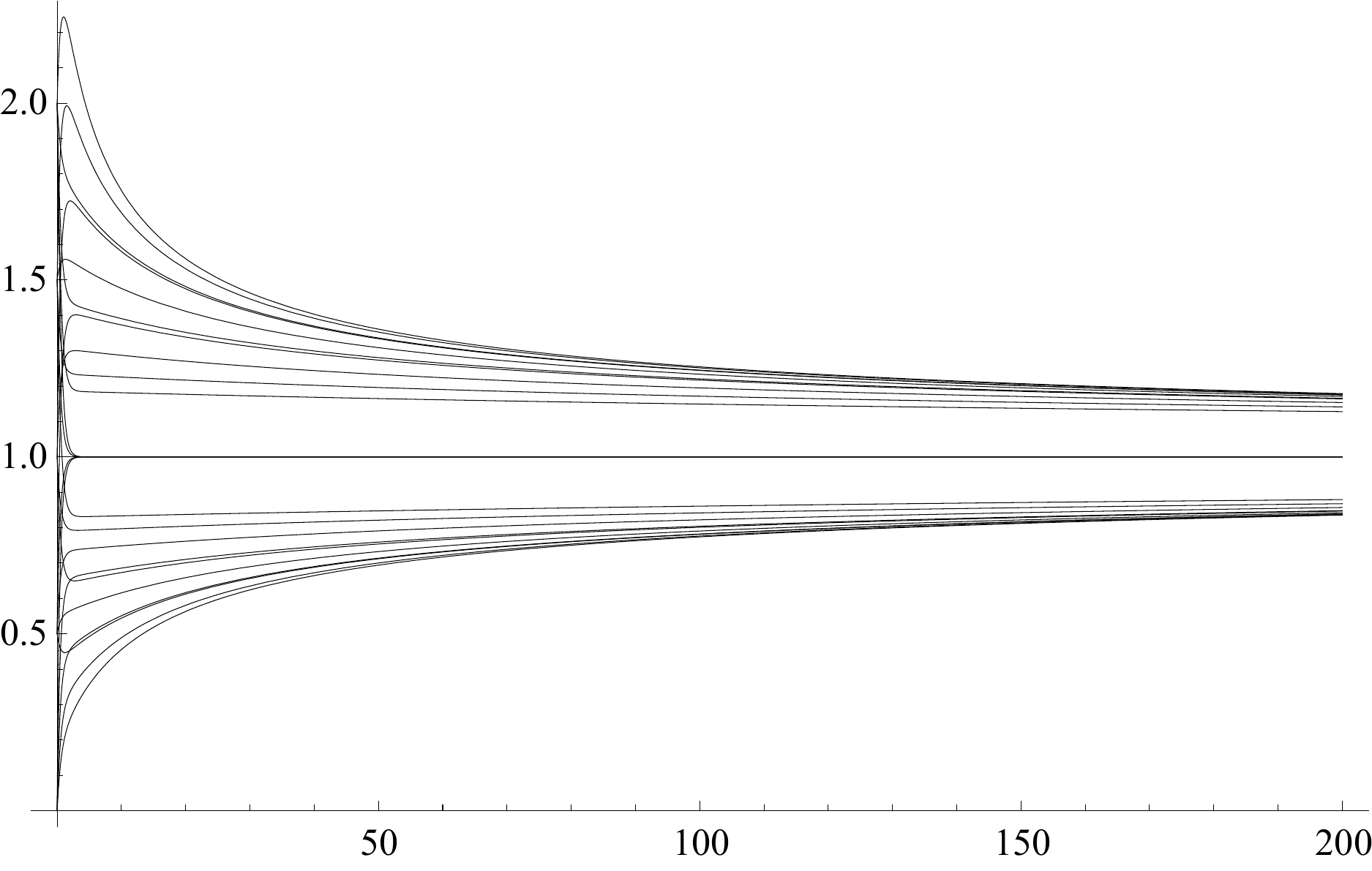}
\caption{$a=e$}
\label{fig:sub2}
\end{subfigure}\\[1ex]
\begin{subfigure}{\linewidth}
\centering
\includegraphics[scale=.35]{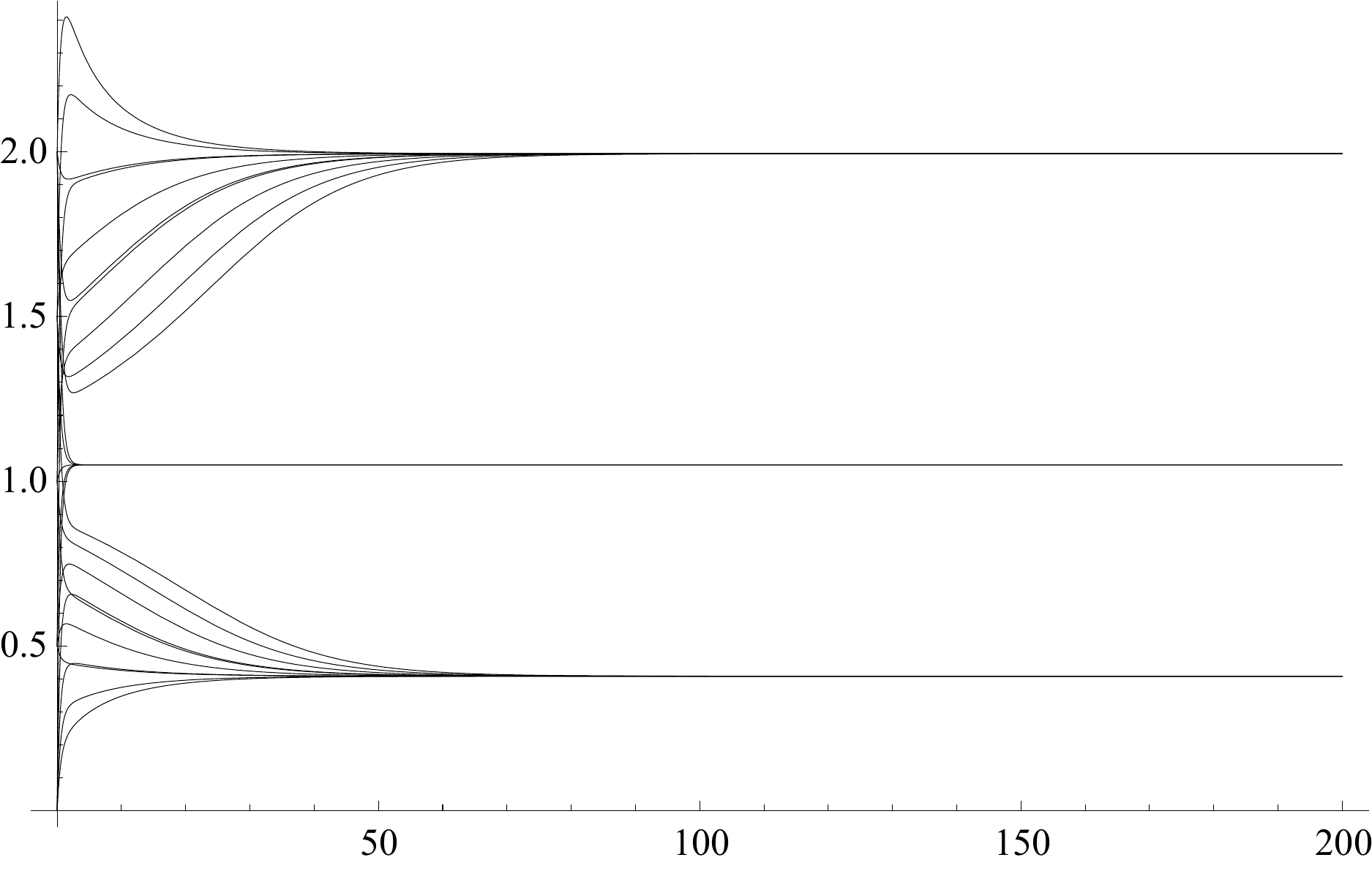}
\caption{$a=3$}
\label{fig:sub3}
\end{subfigure}
\caption{Graphs of $\rho_t^+$ for $\rho_0^\pm\in\{0,0.5,1,1.5,2\}$}
\label{fig:test}
\end{figure}

\begin{figure}[ht]
\begin{subfigure}{.5\linewidth}
\centering
\includegraphics[scale=.35]{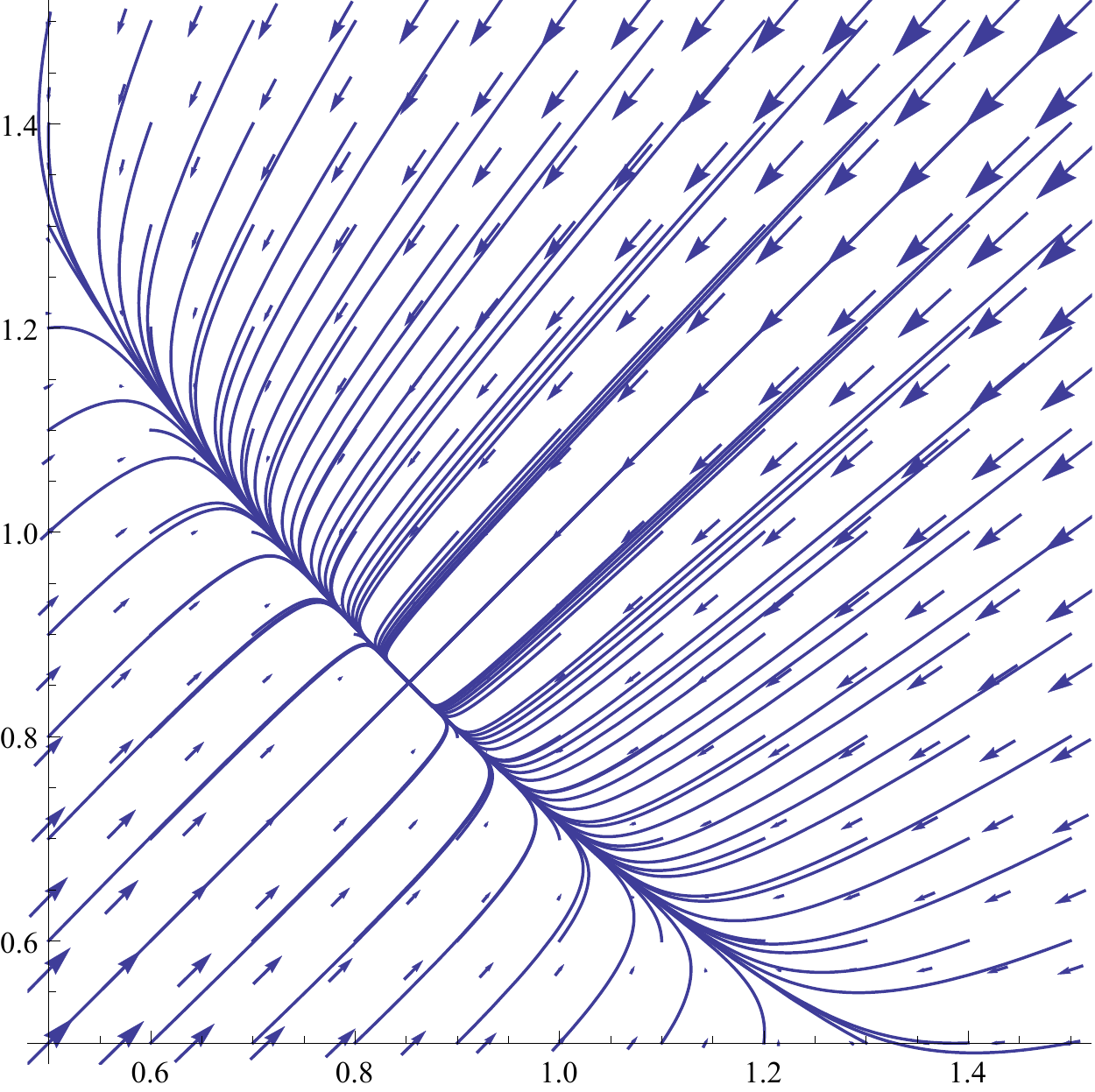}
\caption{$a=2$}
\label{fig2:sub1}
\end{subfigure}%
\begin{subfigure}{.5\linewidth}
\centering
\includegraphics[scale=.35]{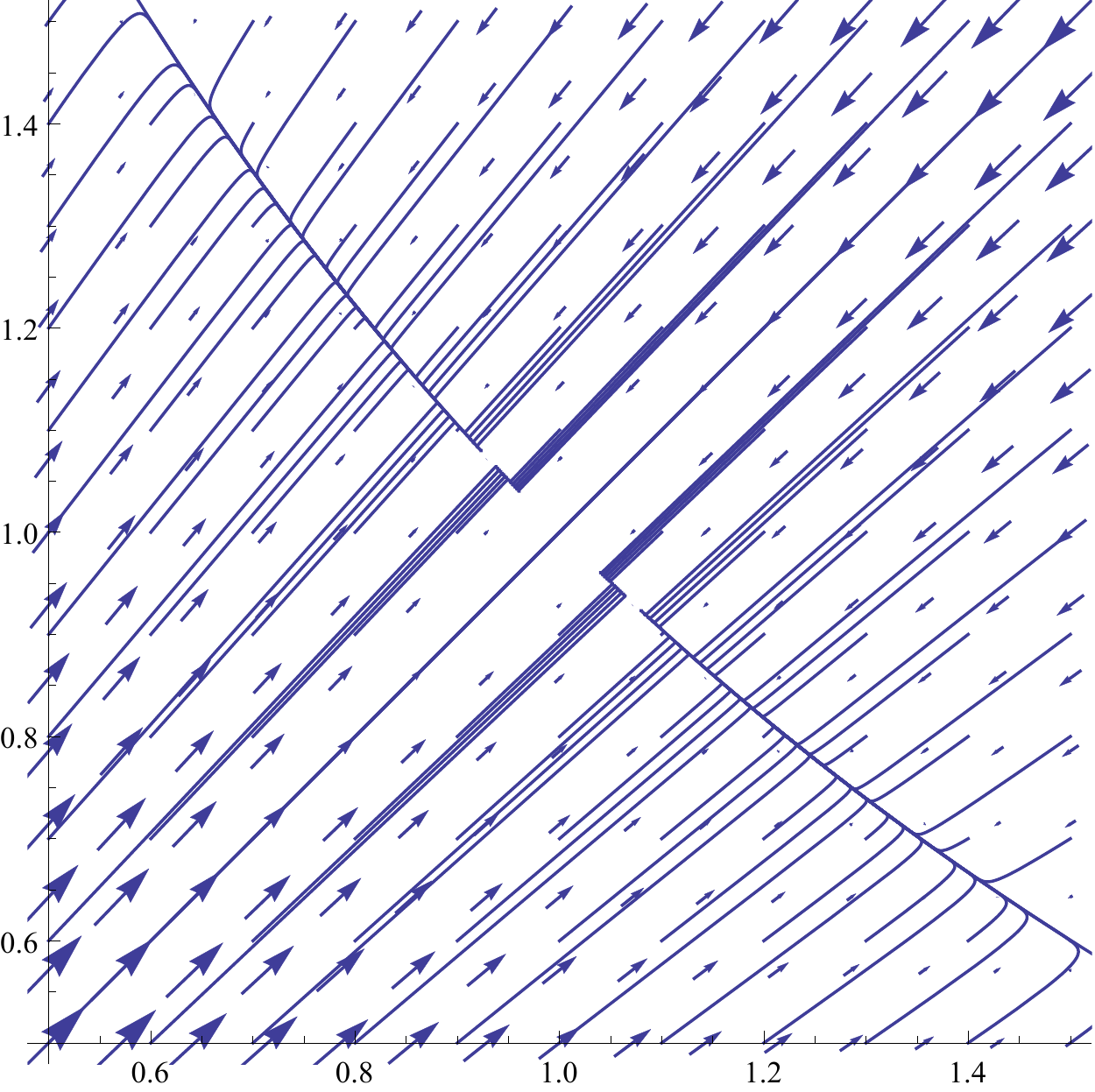}
\caption{$a=e$}
\label{fig2:sub2}
\end{subfigure}\\[1ex]
\begin{subfigure}{\linewidth}
\centering
\includegraphics[scale=.35]{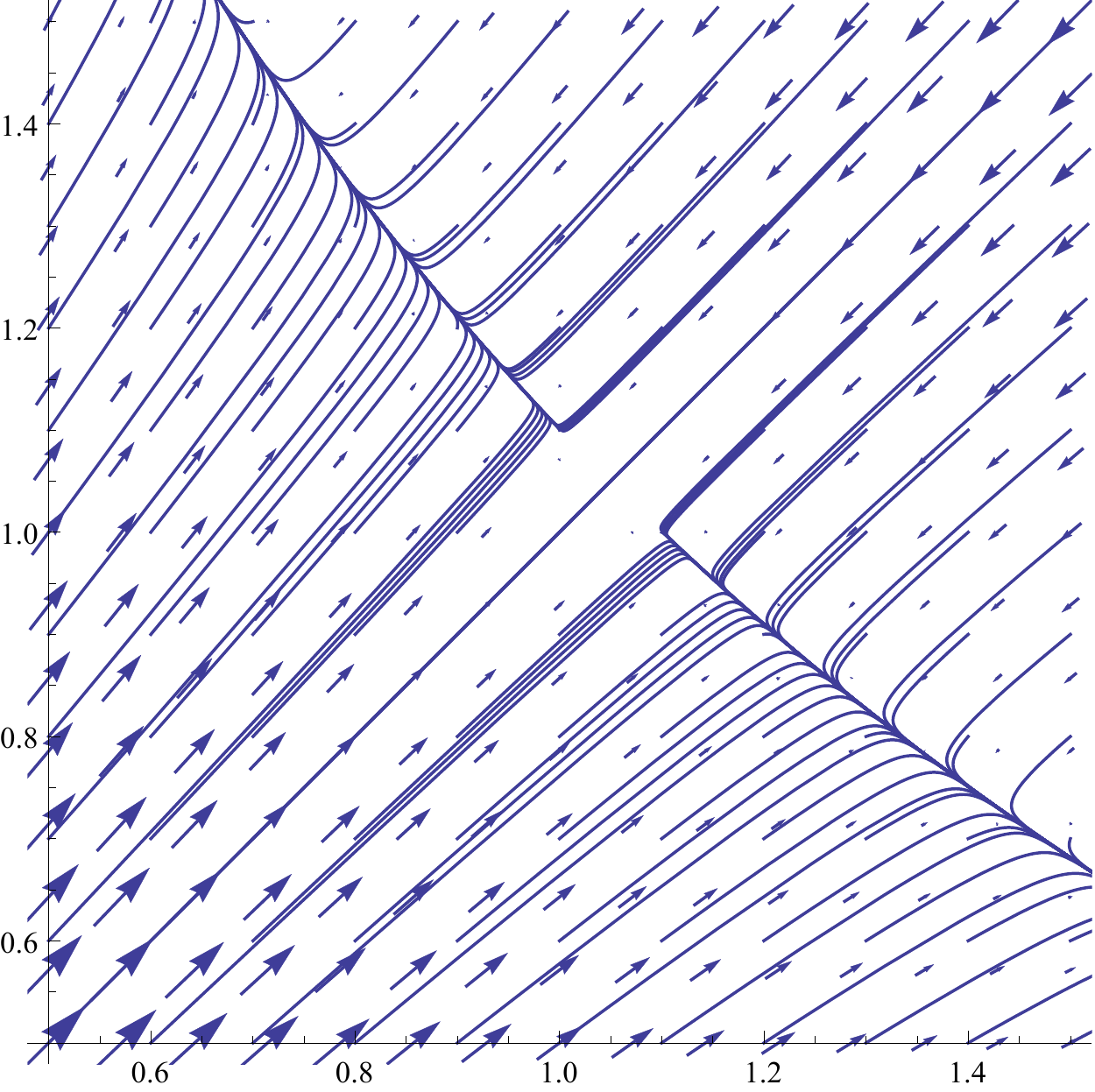}
\caption{$a=3$}
\label{fig2:sub3}
\end{subfigure}
\caption{Phase plane pictures}
\label{fig:test2}
\end{figure}

\begin{theorem}\label{Th1} On a scale of Banach spaces $\{\B_s: 0<s\leq s_0)$
consider the initial value problem
\begin{equation}
\frac{du(t)}{dt}=Au(t),\quad u(0)=u_0\in
\mathbb{B}_{s_0}\label{V1}
\end{equation}
where, for each $s\in(0,s_0)$ fixed and for each pair $s', s''$ such that
$s\leq s'<s''\leq s_0$, $A:\B_{s''}\to\B_{s'}$ is a linear mapping so that
there is an $M>0$ such that for all $u\in\B_{s''}$
\[
\|Au\|_{s'}\leq\frac{M}{s''-s'}\|u\|_{s''}.
\]
Here $M$ is independent of $s',s''$ and $u$, however it might depend
continuously on $s,s_0$.

Then, for each $s\in(0,s_0)$, there is a constant $\delta>0$ (i.e., $\delta=\frac{1}{eM}$)
such that there is a unique function
$u:\bigl[0,\delta(s_0-s)\bigr)\rightarrow\mathbb{B}_s$ which is continuously
differentiable on $\bigl(0,\delta(s_0-s)\bigr)$ in $\mathbb{B}_s$,
$Au\in\mathbb{B}_s$, and solves \eqref{V1} in the time-interval
$0\leq t<\delta(s_0-s)$.
\end{theorem}

\begin{theorem}\label{Thconv} On a scale of Banach spaces
$\{\B_s: 0<s\leq s_0)$ consider a family of initial value problems
\begin{equation}
\frac{du_\eps(t)}{dt}=A_{\eps}u_\eps(t),\quad u_\eps(0)=u_{\eps}\in
\mathbb{B}_{s_0},\quad \eps\geq0,\label{V1eps}
\end{equation}
where, for each $s\in(0,s_0)$ fixed and for each pair $s', s''$ such that
$s\leq s'<s''\leq s_0$, $A_\eps:\B_{s''}\to\B_{s'}$ is a linear mapping so that
there is an $M>0$ such that for all $u\in\B_{s''}$
\[
\|A_{\eps}u\|_{s'}\leq\frac{M}{s''-s'}\|u\|_{s''}.
\]
Here $M$ is independent of $\eps, s',s''$ and $u$, however it might depend
continuously on $s,s_0$. Assume that there is a $p\in\N$ and for each
$\eps>0$ there is an $N_\eps>0$ such that for each pair $s', s''$,
$s\leq s'<s''\leq s_0$, and all $u\in\B_{s''}$
\[
\|A_{\eps}u-A_0u\|_{s'}\leq \sum_{k=1}^p\frac{N_\eps}{(s''-s')^k}\|u\|_{s''}.
\]
In addition, assume that $\lim_{\eps\rightarrow0}N_\eps=0$ and
$\lim_{\eps\rightarrow0}\|u_{\eps}(0)-u_{0}(0)\|_{s_0}=0$.

Then, for each $s\in(0,s_0)$, there is a constant $\delta>0$ (i.e., $\delta=\frac{1}{eM}$)
such that there is a unique solution
$u_\eps:\left[0,\delta(s_0-s)\right)\to\B_s$, $\eps\geq0$, to each initial
value problem (\ref{V1eps}) and for all $t\in\left[0,\delta(s_0-s)\right)$ we
have
\[
\lim_{\eps\rightarrow0}\|u_{\eps}(t)-u_{0}(t)\|_{s}=0.
\]
\end{theorem}


\def\cprime{$'$}

\end{document}